\documentclass{article}

\usepackage[preprint]{neurips_2025}
\usepackage[ruled,vlined]{algorithm2e}
\SetKw{Continue}{continue}
\SetKw{KwParam}{Parameters:}
\usepackage{graphicx}

\usepackage{graphicx}
\usepackage{subcaption}
\usepackage[utf8]{inputenc} 
\usepackage[T1]{fontenc}    
\usepackage{hyperref}       
\usepackage{url}            
\usepackage{booktabs}       
\usepackage{amsfonts}       
\usepackage{nicefrac}       
\usepackage{microtype}      
\usepackage{xcolor}         
\usepackage{natbib}
\usepackage{bm}
\usepackage{bbm}
\usepackage{amsmath}
\allowdisplaybreaks[1]
\usepackage{amsthm}
\newtheorem{definition}{Definition}
\newtheorem{theorem}{Theorem}

\newtheorem{lemma}{Lemma}

\newcommand{\bX}{\bm{X}}
\newcommand{\blam}{\bm{\lambda}}

\newcommand{\bL}{\bm{L}}

\title{Tail-Optimized Caching for LLM Inference}

\author{%
  Wenxin Zhang
  \\
  Columbia Business School\\
  \texttt{wz2574@columbia.edu} \\
  \And
  Yueying Li\\
  Cornell University, Department of Computer Science \\
  \texttt{yl3469@cornell.edu} \\
  \texttt{}\\
  \And 
  Ciamac C. Moallemi \\
  Columbia Business School\\
  \texttt{ciamac@gsb.columbia.edu} \\
  \And
  Tianyi Peng \\
  Columbia Business School\\
  \texttt{tp2845@columbia.edu} \\
}

\begin{document}

\maketitle
\begin{abstract}
Prompt caching is critical for reducing latency and cost in LLM inference---OpenAI and Anthropic report up to 50–90\% cost savings through prompt reuse. Despite its widespread success, little is known about what constitutes an optimal prompt caching policy, particularly when optimizing tail latency—a metric of central importance to practitioners. The widely used Least Recently Used (LRU) policy can perform arbitrarily poor on this metric, as it is oblivious to the heterogeneity of conversation lengths. To address this gap, we propose Tail-Optimized LRU, a simple two-line modification that reallocates KV cache capacity to prioritize high-latency conversations by evicting cache entries that are unlikely to affect future turns. Though the implementation is simple, we prove its optimality under a natural stochastic model of conversation dynamics, providing the first theoretical justification for LRU in this setting---a result that may be of independent interest to the caching community. 
Experimentally,  on real conversation data WildChat~\citep{zhao2024wildchat}, Tail-Optimized LRU achieves up to 27.5\% reduction in P90 tail Time to First Token latency and 23.9\% in P95 tail latency compared to LRU, along with up to 38.9\% decrease in SLO violations of 200ms. 
We believe this provides a practical and theoretically grounded option for practitioners seeking to optimize tail latency in real-world LLM deployments.
\end{abstract}

\section{Introduction}
\noindent\textbf{Prompt Caching is Essential.}
AI capabilities have exploded in recent years, and so has the demand. By December 2024, ChatGPT handled \textit{1 billion} user messages every day with \textit{300 million} weekly active users~\citep{OpenAI2024Stats}. 
To efficiently use scarce and costly GPU resources, prompt caching, or prefix caching, was proposed~\citep{gim2024prompt}: it caches the KV cache of existing queries, allowing a new query to skip some computation by reusing the KV cache it shares with existing queries~\citep{vLLM2025PrefixCaching}. 
Prompt caching can reduce prefill computation thus Time to First Token (TTFT).
This technique has been adopted by OpenAI and Anthropic, both reporting a significant amount (50-90\%) of latency and cost reductions~\citep{OpenAI2024PromptCaching,Anthropic2024PromptCaching}.
Despite the practical impact of prompt caching, little is known about how much existing caching policies—such as the Least Recently Used (LRU) policy—can be improved upon with respect to \textit{key metrics} in LLM inference systems, which is the main motivation of this work.  


\noindent\textbf{Challenges of Optimizing Tail Latency.}
One of such key metrics is \textit{tail latency}. In real‑time user‑facing applications, companies care about high-percentile response time, e.g., 95\% of requests complete within 200 ms. 
In LLM conversation-based applications, users arrive to request services through an alternating sequence of prompts and responses that we call \emph{turns}. Each prompt, along with all previous conversation history, is treated as a job \emph{request}.
When the cache is full, the server must decide \emph{which KV cache blocks to evict}, under four layers of uncertainty: 1) when new conversations are arriving; 2) the number of future turns of existing conversations; 3) the size of future user prompt and model response; and 4) the arrival order of turns from concurrent conversations competing for cache space. 
Here, KV cache blocks are the atomic cacheable units of tokens, e.g., a single block may consist of 128 tokens~\citep{OpenAI2024PromptCaching}.
These intertwined dynamics make tail latency optimization in LLM inference uniquely challenging.


\noindent\textbf{Existing Approaches.}
Classic caching/paging policies 
mostly focus on maximizing \textit{cache hit rate}, not \textit{tail latency}. Among these, the LRU policy is perhaps the most representative, evicting the cache item that was accessed least recently. LRU has been widely adopted in LLM inference systems, including vLLM~\citep{kwon2023efficient}, SGLang~\citep{zheng2024sglang}, and Mooncake~\citep{qin2025mooncake}. However, LRU does not account for the fact that different blocks within a request may have varying effects on tail latency, thus leaving room for further optimization. See Figure~\ref{fig:LRU-example} for an illustrative example.

\begin{figure}
    \centering
    \includegraphics[width=0.8\linewidth]{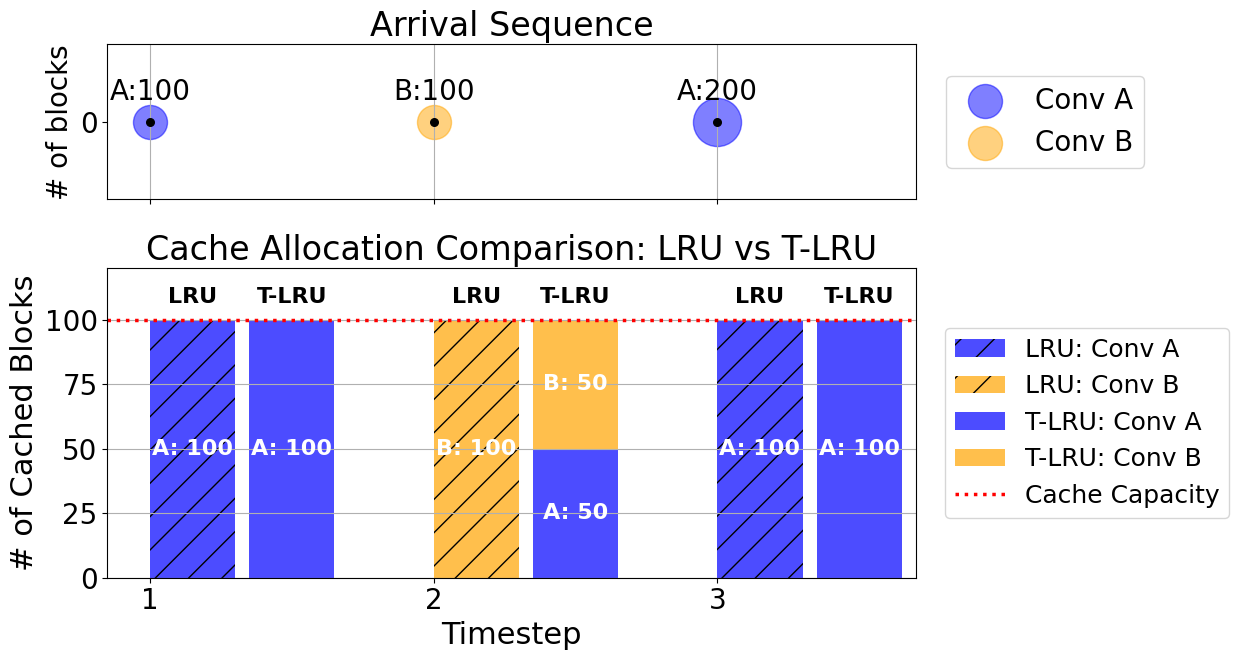}
  \caption{\textbf{Failure of LRU on Tail Latency.} 
Consider an example with two conversations (A and B) and a total of three turns (A has two turns; B has one). The top panel shows the total job size (conversation history plus user prompt) at each step. For simplicity, we assume the response length is zero and that each user prompt consists of 100 blocks. The bottom panel shows the number of cached blocks updated after each turn. The cache capacity is 100 blocks. We measure tail latency as the maximum processing time among the three requests (i.e., the $66.7\%$-th percentile). Under LRU, all of conversation A's cache blocks are evicted immediately after step 2, resulting in a maximum of 200 uncached blocks at step 3. Therefore, the tail latency under LRU corresponds to the time required to process 200 blocks. In contrast, if at step 2 we partially evict A's cache so that both A and B retain 50 cached tokens each, then regardless of whether the third request comes from A or B, the maximum number of uncached blocks is reduced to 150—an improvement of 25\%. This improvement is what our proposed policy, T-LRU, is designed to achieve.}
    \label{fig:LRU-example}
\end{figure}

\noindent\textbf{Our Approach: Tail-Optimized LRU (T-LRU).} 
In this work, we focus on optimizing the tail latency of TTFT (Time to First Token) through improved cache eviction policies. To this end, we introduce the metric \textit{Tail Excess Latency} (TEL):
\begin{equation}
    \text{TEL} = \sum_{i} \max\{\text{TTFT}(i) - \xi_{s},\ 0\},
\end{equation}
where $\text{TTFT}(i)$ denotes the Time to First Token for the $i$-th request, and $\xi_{s} \geq 0$ is a user-specific latency threshold (e.g., $\xi_{s} = 200$\,ms). TEL captures the total TTFT exceeding the threshold, serving as a practical proxy for tail latency while treating each request independently. See Section~\ref{sec: background} for further discussion of this objective.

Empirically, TTFT is known to be approximately linear in the number of uncached tokens\footnote{This approximation holds well when the uncached length is not excessively long and the prefill of the request is not batched with others. We use this approximation for analytical purposes, while the actual TTFT is measured in our experiments.}, i.e.,
\begin{equation}\label{eq:linear-TTFT}
    \text{TTFT}(i) \approx \alpha \cdot b(i),
\end{equation}
where $\alpha$ is a constant determined by the model and GPU configuration, and $b(i)$ is the number of uncached blocks for request $i$ (see~\citep{horton2024kv} and Figure~\ref{fig:linear-uncached-token}). We use this approximation to motivate our policy design and simplify the analysis. By defining $\xi = \xi_{s} / \alpha$, TEL can be equivalently expressed as
\begin{equation}
    \text{TEL} \approx \alpha \sum_{i} \max\{b(i) - \xi,\ 0\}.
\end{equation}
Our intuition is the following: for a conversation turn with conversation history length $L$ and whose next prompt is expected to add $Q$ blocks, caching more than $L + Q-\xi$ blocks for this turn cannot improve TEL as the number of uncached blocks is already lower than $\xi$. Any blocks beyond this TEL-safe budget can therefore be evicted ``for free". See Figure \ref{fig:proactive_trimming} for illustration of this idea. 
Motivated by this observation, we propose Tail-Optimized LRU (T-LRU), a simple modification of LRU. Upon cache overflow, T-LRU works in two phases: 
\begin{enumerate}
    \item TEL-safe trimming: first evict blocks from conversations whose cache size exceeds the TEL-safe budget $L+Q-\xi$;
    \item LRU-as-usual: If space is still needed, evict blocks using LRU.
\end{enumerate} 
In Figure \ref{fig:LRU-example}, the performance of T-LRU with $\xi=150$ is shown. In practice, the next-prompt length is unknown; T-LRU can use a constant $\hat{Q}$ such as the empirical average for estimating the length of the next-prompt.\footnote{Furthermore, $Q - \xi$ can be combined and interpreted as one tunable parameter that determines how many blocks at the end of the conversation history are considered safe to evict.} Implementation requires only one extra bookkeeping: mark TEL-safe blocks as ``infinitely old,'' after which any existing LRU engine can evict them as usual.

\begin{figure}[!htbp]
  \centering
    \includegraphics[width=0.6\linewidth]{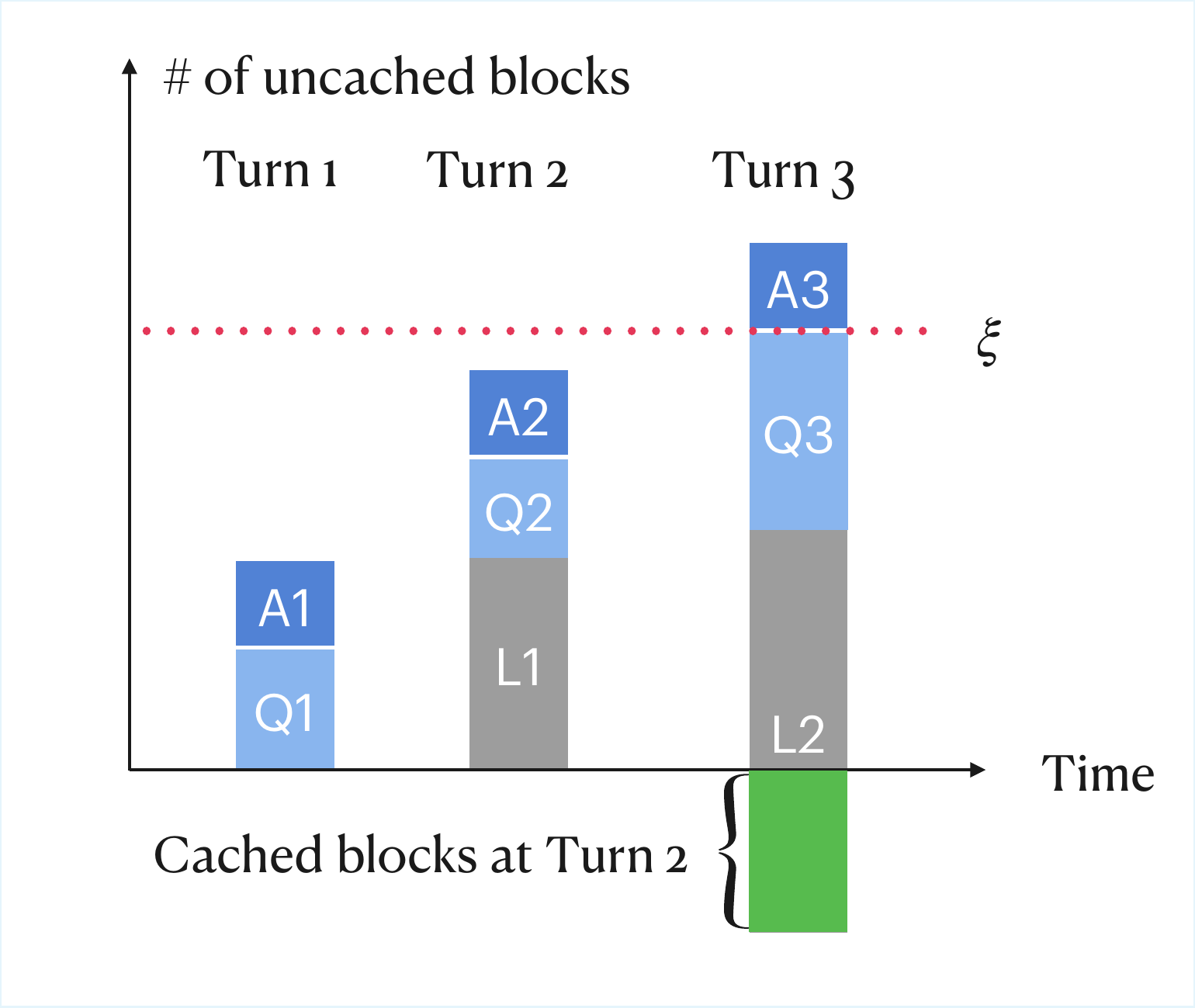}
  \caption{A three-turn conversation $\{(Q_1,A_1),(Q_2,A_2),(Q_3,A_3)\}$ is shown, $L_i$ denotes conversation history (i.e., $L_1=Q_1+A_1, L_2=L_1+Q_2 + A_2$). Bars indicate the number of uncached blocks each turn incurs; the red dashed line marks the latency threshold $\xi$. Turn 1: projected load for next request is $L_1 + Q_2 <\xi$; no caching is needed. Turn 2: now $L_2 + Q_3 >\xi$, at least $L_2 + Q_3 -\xi$ blocks must be cached (green), but caching more won't improve TEL further.
  }
    \label{fig:proactive_trimming}
\end{figure}

\textbf{Theoretical Optimality.} On the theoretical front, we establish the following results:
\begin{enumerate}
    \item Belady’s algorithm~\citep{belady1966study}, which evicts the block whose next use is furthest in the future, is a clairvoyant caching strategy known to achieve the \textit{hindsight-optimal hit rate}. We show that combining \emph{TEL-safe trimming} with Belady’s algorithm yields an policy being \textit{hindsight optimal for TEL}  (Theorem~\ref{thm: modified baledy is optimal}). This result justifies the use of T-LRU, as LRU is a widely used heuristic for Belady’s algorithm.
    \item In the caching literature, the settings under which LRU is \textit{online optimal} 
    are not well understood. To address this, we introduce a novel stochastic model for multi-turn conversations that captures both uncertainty and temporal locality in LLM workloads. Within this framework, we prove that a generalized T-LRU is optimal (Theorem~\ref{thm: optimality of tail LRU rand}), incorporating the optimality of classical LRU (for average latency) and T-LRU as special cases. \textit{To the best of our knowledge, this is the first result demonstrating the optimality of LRU under a natural, conversation-driven arrival model, and may be of broad interest to the caching community.}
\end{enumerate}

\textbf{Strong Practical Performance.} Finally, we evaluate the performance of T-LRU and LRU on real multi-turn chat traces from ShareGPT~\citep{sharegpt} and WildChat~\citep{zhao2024wildchat}. Our policy achieves up to a 23.9\% reduction in P95 tail latency compared to LRU, and up to a 38.9\% decrease in SLO violations relative to the strongest baseline. We also provide insights into selecting the latency threshold $\xi$ (Section~\ref{sec: experiment}). These results highlight the practical advantages of T-LRU.

The rest of the paper is organized as follows. We discuss related work in Section~\ref{sec: related work}. In Section~\ref{sec: background}, we describe the problem setting and present the hindsight-optimal policy for TEL. Section~\ref{sec:algorithm} introduces the T-LRU policy and demonstrates that it is never worse than LRU for optimizing TEL. In Section~\ref{sec: stochastic model}, we prove the optimality of T-LRU for a class of stochastic models that capture conversational arrivals. Section~\ref{sec: experiment} presents the experimental results. Finally, Section \ref{sec:conclusion} discuss future directions.

\section{Related Work}\label{sec: related work}

\noindent\textbf{Classic Caching and Paging.}
The caching problem has been extensively studied from a competitive analysis perspective, with binary cache hit/miss metrics. Foundational results include the hindsight-optimal Belady's algorithm~\citep{belady1966study}, optimal deterministic online policies~\citep{sleator1985amortized}, and randomized algorithms~\citep{fiat1991competitive}. Beyond worst-case competitive ratios, many works model structured request processes to capture temporal locality observed in practice, including access graph models~\citep{borodin1995competitive,chrobak1999lru}, independent reference models~\citep{coffman1973operating,king1972analysis,dan1990approximate,che2002hierarchical}, Shot Noise models~\citep{leonardi2015least}, LRU Stack and Working Set models~\citep{jain1990characteristics}, and Inter-Reference Gap models~\citep{phalke1995inter}. For broader overviews, see~\cite{borodin2005online}; for caching with predictions, see~\cite{lykouris2021competitive,mhaisen2022optimistic}.

\noindent\textbf{Alternative Metrics.}
Beyond binary hit/miss metrics, several systems optimize continuous objectives. Darwin~\citep{chen2023darwin} balances hit rate and disk writes in CDN caching. RobinHood~\citep{berger2018robinhood} and SNC-Meister~\citep{zhu2016snc} target tail latency in web services and multi-tenant systems, respectively. AdaptSize~\citep{berger2017adaptsize} handles variable-size objects in content delivery networks. However, these works focus on traditional web caching rather than the unique characteristics of LLM inference: multi-turn conversations with growing context windows and strict real-time latency requirements.

\noindent\textbf{Prompt Caching for LLM Inference.}
Prompt caching reuses precomputed KV states from conversation history to reduce prefill computation. Recent systems explore various aspects: prefix caching~\citep{kwon2023efficient}, attention reuse~\citep{gim2024prompt}, RAG optimization~\citep{jin2024ragcache}, structured generation~\citep{zheng2024sglang}, load balancing~\citep{srivatsa2024preble}, hierarchical storage~\citep{qin2025mooncake}, and optimal model multiplexing~\citep{zhu2023optimal}. While these works demonstrate the value of prompt caching, they primarily optimize throughput or average latency rather than tail latency.

Recent work on LLM serving has begun addressing tail latency through architectural innovations. Sarathi-Serve~\citep{agrawal2024taming} uses chunked prefill and stall-free scheduling, while DistServe~\citep{zhong2024distserve} disaggregates prefill and decoding for goodput optimization. These systems target tail latency but do not use caching as the primary optimization lever.

\noindent\textbf{Our work.}
To sum up, existing works have studied continuous metrics other than binary cache hit/miss and variable-size objects, but we are the first to use prompt caching as the primary lever to optimise TTFT tail latency, and the first to provide theoretical guarantees for tail-excess latency. Our work bridges the gap between classical caching theory, modern LLM serving systems, and tail latency optimization.

\section{Hindsight Optimal Policy}\label{sec: background}

\noindent\textbf{Deterministic Arrival Trace.} 
To begin, we describe the prompt caching problem for LLM inference under a deterministic arrival trace to study the hindsight optimality, deferring the stochastic model to Section~\ref{sec: stochastic model} for online optimality. Consider a discrete time horizon $t = 1, 2, \dotsc, T$ with $N$ conversations over $T$ steps. For each conversation $i \in [N]$, let $\mathcal{T}_i \subseteq [T]$ denote the set of time steps when conversation $i$ issues a request. Assume $\mathcal{T}_i$ is disjoint from each other. For each $t \in \mathcal{T}_i$, let $q_{i,t}$ and $a_{i,t}$ denote the lengths of the user prompt and model response for that turn, respectively, measured in blocks; otherwise, set $q_{i,t} = a_{i,t} = 0$ for $t \notin \mathcal{T}_i$.

\noindent\textbf{Caching State.}
For $i\in [N]$ and $t\in [T]$, let $x_{i,t}$ denote the number of cached blocks for conversation $i$ at time $t$, \textit{before the request arrival}.
The caching state $\{x_{i,t}\}$ must satisfy the following constraints.

First, the total number of cached blocks cannot exceed the cache capacity $C$ at any time:
\begin{align}
    \sum_{i\in[N]} x_{i,t} \leq C, \qquad \forall t\in[T] \quad \text{(capacity constraint)}
    \label{cons: capacity}
\end{align}

Second, for each conversation, the number of cached blocks cannot exceed the length of its conversation history up to time $t$:
\begin{align}
    x_{i,t+1} \leq \sum_{j=1}^{t}(q_{i,j} + a_{i,j}),
    \qquad \forall i\in[N],\; t\in\mathcal{T}_i,\; t < T
    \label{cons: optional caching}
\end{align}

Third, the cache allocation for a conversation can only increase when a request from that conversation arrives; otherwise, it can only decrease or stay the same:
\begin{align}
    x_{i,t+1} \leq x_{i,t},
    \qquad \forall i\in[N],\; t\notin\mathcal{T}_i,\; t < T
    \quad\text{(cache increases only on arrival)}
    \label{cons: cache on arrival}
\end{align}

A caching policy must determine the caching state $\{x_{i,t}\}$ subject to these constraints. When a request arrives, the server can reuse any cached blocks in $\{x_{i,t}\}$, while any evicted blocks must be recomputed.\footnote{In practice, evicted blocks may be retrieved from other storage layers (e.g., CPU DRAM or SSD). Here, we focus on a single-layer cache and do not consider inter-layer transfers.}

For simplicity, we assume \textit{optional caching}: a request is not required to be added to the cache after serving. Extending to \textit{forced caching} is straightforward (by replacing the inequality in constraint~\eqref{cons: optional caching} with equality); we defer this discussion to Appendix~\ref{apx: forced caching}.

\noindent\textbf{Latency Objective: TEL.}  
We aim to optimize the tail latency metric via the caching policy. Percentile-based tail latency is notoriously difficult to optimize directly, so we introduce the \emph{Tail Excess Latency} (TEL) metric as a tractable surrogate:
\begin{align*}
    \text{TEL} := \sum_{i\in [N],\, t\in \mathcal{T}_i} \max\{\text{TTFT}(i, t) - \xi_{s},\, 0\},
\end{align*}
where $\text{TTFT}(i, t)$ denotes the Time to First Token for the request at time $t$ of conversation $i$ and $\xi_{s}$ is a pre-defined threshold. TEL is analogous to Conditional Value at Risk (CVaR), but with an explicit, user-specified threshold~\citep{bauerle2011markov,chow2015risk}; it also resembles SLO attainment metrics~\citep{zhong2024distserve}, with the distinction that TEL penalizes larger violations proportionally more. Setting $\xi_{s} = 0$ recovers the average-latency objective. While our theoretical analysis centers on TEL, our experiments report conventional metrics (tail latency and SLO violation rate) for straightforward comparison in future studies.

As discussed in~\eqref{eq:linear-TTFT}, to simplify the analysis and guide policy design, we assume TTFT grows linearly with the number of uncached blocks—an assumption empirically justified (see Figure~\ref{fig:linear-uncached-token}) by the dominance of feed-forward linear layers during inference~\citep{kamath2025pod,zhu2024nanoflow,ye2025flashinfer}. Thus,
\begin{align*}
\text{TEL} = \alpha \sum_{i \in [N],\, t \in \mathcal{T}_i} \max\left\{\sum_{j=1}^{t} q_{i,j} + \sum_{j=1}^{t-1} a_{i,j} - x_{i,t} - \xi,\, 0\right\}
\end{align*}
where $\xi = \xi_{s} / \alpha$, and $\sum_{j=1}^{t} q_{i,j} + \sum_{j=1}^{t-1} a_{i,j} - x_{i,t}$ is the number of blocks that need to be computed the request of conversation $i$ at time $t$ (here, $\sum_{j=1}^{t-1} q_{i,j} + \sum_{j=1}^{t-1} a_{i,j}$ is the total number of blocks in the conversation history, and $q_{i,t}$ corresponds to the new request prompt at time $t$).

\begin{figure}[!htbp]    \centering\includegraphics[width=0.6\linewidth]{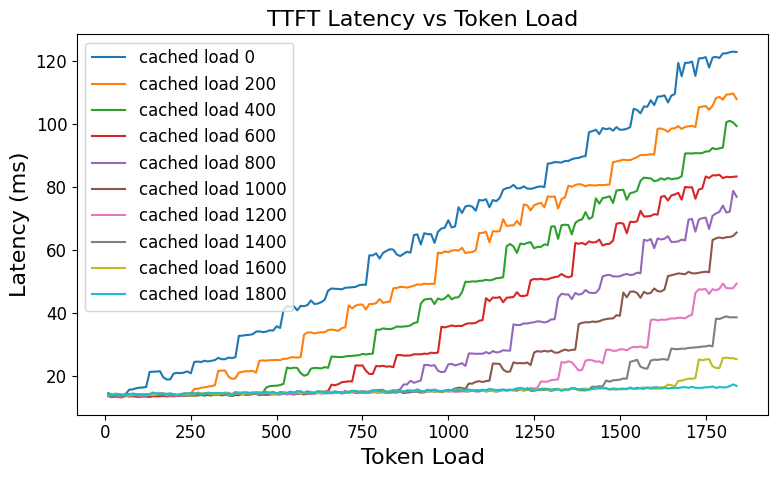}
        \caption{Experimental results demonstrate that TTFT increases approximately linearly with the number of uncached tokens. The plot shows TTFT latency as a function of the total prompt length (Token Load) and the size of the cached prefix (Cached Load). We conducted our experiments using Vicuna-7B with vLLM’s prefix caching enabled on a Colab A100 GPU.}
    \label{fig:linear-uncached-token}
\end{figure}

\noindent\textbf{Hindsight Optimal Policy for TEL}
With both the decision timeline and objective clarified, we next ask: \emph{If the system knew the entire future arrival trace, what caching policy would minimize TEL?} Understanding this hindsight-optimal benchmark unveils what's important to improve tail latency and thus guides our policy design.

The hindsight optimal policy chooses cache variables $\{x_{i,t}\} \in \mathbb{N}$: the number of cache blocks conversation $i$ can reuse at the beginning of step $t$, and slack variables $u_{i,t}\geq 0$ to minimize TEL:
\begin{align}
  \min_{x_{i,t},u_{i,t} \in \mathbb{N}}\;&\sum_{i\in[N]}\sum_{t\in\mathcal{T}_i} u_{i,t}
  \label{eq: opt for min TEL}\\
  \text{s.t. } \; & \eqref{cons: capacity}, \eqref{cons: optional caching}, \eqref{cons: cache on arrival} \notag\\
  &u_{i,t}\ge\sum_{j=1}^t q_{i,j}+\sum_{j=1}^{t-1}a_{i,j}-x_{i,t}-\xi,
  &&\forall i\in[N],\;t\in\mathcal{T}_i \quad\text{(slack variable)} \notag
\end{align}

Here $u_{i,t}$ describes the TEL objective as the $u_{i,t}$ equals $(\sum_{j=1}^t q_{i,j} + \sum_{j=1}^{t-1} a_{i,j} - x_{i,t} -\xi)^+$ in the optimal solution.

\begin{theorem}[Hindsight-Optimal Policy Structure]\label{thm: modified baledy is optimal}
The hindsight-optimal policy that minimizes TEL restricts the number of KV-cache blocks allocated to each conversation within the TEL-safe budget, given by $\left(\sum_{j=1}^{t} q_{i,j} + \sum_{j=1}^{t-1} a_{i,j} - \xi\right)^{+}$. If the total allocation exceeds the cache capacity, the policy evicts cache blocks from conversations whose next requests are expected to arrive furthest in the future.
\end{theorem}

See Appendix \ref{apx: hindsight optimal policy} for detailed proof.
Setting $\xi = 0$ recovers average latency minimization, and in this case, the theorem implies the furthest-in-future eviction strategy is optimal---we recover the Belady optimal policy, which was shown to maximize cache hit rate in classical caching problems~\citep{belady1966study}.
Therefore, the hindsight optimal policy for minimizing TEL can be characterized as a threshold-capped version of the Belady policy.
We call this policy Tail-Optimized Belady.

\section{Tail-Optimized LRU}\label{sec:algorithm}
Inspired by the Tail-Optimized Belady policy (Theorem \ref{thm: modified baledy is optimal}), we propose Tail-Optimized LRU, an online policy that also maps into the future and caches \emph{just enough} to prevent conversations' next turns from affecting \text{TEL}. 
To adapt the hindsight optimal policy to online settings, we make two modifications: 1) replace the actual next-prompt length $Q$ with an estimate $\hat{Q}$; 2) use the least-recently-used eviction strategy rather than the furthest-in-future.

Pseudocode is given in \ref{alg:tail-optimized-lru}. We believe Tail-Optimized LRU is a practical, low-friction upgrade for any LLM caching system that already relies on LRU.

{\scriptsize
\begin{algorithm}[!htbp]
\caption{Tail‑Optimized LRU Policy}
\label{alg:tail-optimized-lru}
\SetAlgoLined
\SetAlgoNoEnd
\DontPrintSemicolon
\KwIn{Number of conversations $N$, Timestamp of Last Turn $\{\tau_i\}$, Number of cached blocks $\{X_i\}$, conversation history lengths $\{L_i\}$, arriving conversation $\theta$, arriving conversation length $L'_\theta$} 
\KwParam{Policy parameters: threshold $\xi$, next‑turn length estimate $\{\hat{Q}_i\}$}

\KwOut{Updated cache sizes $\{X_i\}$} 
$L_\theta \gets L'_\theta, X_\theta \gets L_\theta, \tau_\theta \gets \text{Current Timestamp}$ \tcp*{Update system state for arriving conversation}
\If{$\sum_{i\in [N]} X_i > C$}{ 
  \ForEach{$i\in [N]$}{
    \If{$X_i\ge L_i + \hat{Q}_i - \xi$}{
        $X_i \gets X_i -1$ \tcp*{``Free Eviction'' under TEL objective}
        \If{$\sum_{i\in [N]} X_i \leq C$}{\Return{$X$}}
    }
  }
}
\While{$\sum_{i\in [N]} X_i > C$}{ 
  Find $j = \arg\min_{i \in [N]: X_i \geq 1} \tau_i$ \tcp*{Evict using LRU}
  $X_j\gets X_j-1$\\  
}
\Return{$X$}
\end{algorithm}
}

\noindent\textbf{Lightweight Integration with Existing Caching Systems.}\label{sec: easy adoption}
Implementing Tail-Optimized LRU in a cache system based on LRU only requires an extra bookkeeping: mark blocks that can be evicted ``for free'' (i.e., blocks identified in the proactive trimming phase) as infinitely old. Then these blocks can be evicted seamlessly by any existing LRU engine in the usual way. 
This design is also compatible with paged KV cache technique which stores keys and values in non-contiguous memory space.

\section{Optimality of Tail-Optimized LRU in a Stochastic Conversation Model}\label{sec: stochastic model}
In this section, we prove the optimality of a generalized Tail-Optimized LRU policy under stochastically generated traces, which covers the optimality of T-LRU (Policy \ref{alg:tail-optimized-lru}) and LRU as special cases. 

\noindent\textbf{Stochastic Conversation Model.}
Classic caching models fail to capture the nature of LLM workload: unlike traditional cache systems where objects have static sizes and independent access patterns, LLM workloads consist of multi-turn conversations that dynamically start, evolve, and terminate over time; see related work in Section~\ref{sec: related work}. To address this gap, we build a novel stochastic model that characterizes these unique features of LLM workloads. 

To characterize multi-turn conversations, our model must address four fundamental questions: (1) when do new conversations start? (2) when will an active conversation generate its next prompt? (3) how many tokens appear in prompts and responses? (4) when do conversations terminate?
Real traces from \cite{zhao2024wildchat} reveals strong \emph{temporal locality} in multi-turn conversations: \textit{the longer a user goes without sending their next prompt, the less likely they are to ever return}. 
This observation has a direct practical implication: least-recently-used conversations are also least-likely-to-return. Consequently, when cache capacity constraints force eviction decisions, prioritizing recently active conversations aligns with their probability of future requests.
We capture this pattern by modeling the number of active conversations as a continuous-time birth-death process. Specifically, 
\begin{itemize}
    \item New conversations are ``born" at rate $\lambda_{\textsf{conv}}>0$, and each active conversation ``dies'' at rate $\mu>0$. We index conversations by their arrival order, $i = 1,2,\ldots$.
    \item While active, conversation $i$ generates requests according to an independent Poisson process with rate $\bar{\lambda}_i>0$. 
    \item At each turn, a random prompt length $Q$ is drawn from a known (possibly conversation-specific) distribution; the length of model responses $A$ follows an arbitrary distribution that the decision-maker does not need to know. 
\end{itemize}

This design has an appealing property: conversations with exponential ``death clocks" naturally implement the temporal locality observed empirically. The longer a conversation stays quiet, the more likely it has terminated, making it a safe candidate for cache eviction. Building on this insight, we develop T-LRU, a policy that maintains LRU ordering while prioritizing conversations whose next request would breach a time-to-first-token (TTFT) latency threshold. In this section, we prove that a generalized version of T-LRU that minimizes expected total eviction loss (TEL) is optimal under our stochastic conversation model.

For simplification, we assume that KV caches cannot be reused across different conversations. This assumption is also motivated by security and privacy concerns: e.g., vLLM implement cache isolation to prevent timing-based inference attacks.\footnote{\url{https://docs.vllm.ai/en/stable/design/v1/prefix_caching.html}}

\noindent\textbf{Belief Markov Decision Process.}
The decision-maker only observes the turns as they arrive, but departures are never observed. Thus the problem is modeled as a partially observable Markov decision process (POMDP), where the optimal policy is defined for each possible belief state over the POMDP states.
As the conversations arrive and depart independently, we can decompose the belief state to be the individual expected turn rates of each conversation. 

Let $\pi_i(t)$ denote the decision-maker's belief at time $t$ that conversation $i$ is still active, then its expected turn-arrival rate is $\pi_i(t)\cdot \bar{\lambda}_i$. The belief is updated using 
\[\pi_i(t) = \exp(-\mu(t - \text{last turn time})),\]
as each conversation lasts for an exponential amount of time with mean $\mu$.

Therefore, the system state at time $t$ of the belief MDP is given by $(\bm{\lambda}(t), \bm{L}(t), \bm{X}(t))$, where $L_i(t)$ is the total length (blocks) of conversation $i$ at time $t$ and $X_i(t)$ is the number of KV cache blocks from conversation $i$, all prior to the arrival at time $t$. The decision-maker chooses $\bX'$, which blocks to cache after serving each request, to minimize the Tail Excess Latency for $M$ requests for arbitrary $M \in \mathbb{N}$. Appendix \ref{apx: proof of ETLRU} details the finite‑horizon Bellman equation.

\noindent\textbf{Expected-Tail-Optimized LRU (ET-LRU).}
ET-LRU chooses the post-arrival cache allocation that minimizes the expected TEL at the next turn, using current beliefs of turn-arrival rates. Formally,
\begin{definition}[Expected Tailed-Optimized LRU]\label{def: TLRU rand}
Let $\theta$ denote the index of the conversation arriving at time $t$ with new prompt length $Q$ and model response length $A$, $\bX(t^-)$ denote the cache state \textit{before} arrival, and $\bm{\lambda}(t^+),\bm{L}(t^+)$ denote the belief turn-arrival rates and conversation history lengths updated \textit{after} service ($\lambda_\theta(t^+) = \bar{\lambda}_\theta, L_\theta(t^+) = L_\theta(t^-) +Q + A$).
ET-LRU chooses
{\small
\begin{align}
   \bm{X}^{\textsf{ETLRU}} \in \arg\min & \sum_{i} \lambda_i(t^+) \cdot \mathbb{E}[(L_i(t^+) + Q_i - Y_i -\xi)^+] \label{eq: opt rand new arrival} \\
   \text{s.t. }&  Y_\theta \leq L_\theta(t^+),  \quad\text{(optional caching)}  \label{cons: ETLU optional caching}\\
   & Y_i \leq X_i(t^-), \forall i\neq \theta,  \quad\text{(cannot conjure caches)} \label{cons: ETLU cache upon arrival}\\
   & \sum_{i} Y_i \leq C.  \quad\text{(capacity constraint)} \notag
\end{align}}
where the expectation is taken over $Q_i$, the random user prompt length for conversation $i$ at its next arrival. 
\end{definition}
The probability that conversation $i$ generates the next request is proportional to its belief turn-arrival rates $\lambda_i(t^+)$, thus $\lambda_i(t^+) \cdot \mathbb{E}[(L_i(t^+) + Q_i - X_i'-\xi)^+]$ is the expected TEL contributed by conversation $i$. 
Here constraint \eqref{cons: ETLU optional caching} implies that the decision-maker can cache at most the total conversation history of conversation $\theta$ just served; constraint \eqref{cons: ETLU cache upon arrival} says a conversation's cache allocation can grow only whenn it arrives.

\begin{theorem}[Optimality of Expected-Tail-Optimized LRU]\label{thm: optimality of tail LRU rand}
Expected-Tail-Optimized LRU (Definition \ref{def: TLRU rand}) is an optimal online caching policy for minimizing Tail Excess Latency under the stochastic conversation model above.
\end{theorem}
The proof is by induction on the number of turns and comparing the value-to-go functions under our policy and another policy that evicts differently. 
Here, the least-recently-used time reflects the expected turn-arrival rate and thus can approximate furthest-in-future.
Theorem \ref{thm: optimality of tail LRU rand} has three implications:
\begin{itemize}
    \item LRU is optimal for average latency. Setting $\xi = 0$ and assuming homogeneous turn-arrival rates across conversations, the objective in optimization problem \eqref{eq: opt rand new arrival} reduces to $\max \sum_{i} \exp(-\mu(t - \text{last turn time}))Y_i$, thus ET-LRU reduces to LRU. Therefore, Theorem \ref{thm: optimality of tail LRU rand} establishes the optimality of LRU for minimizing average latency under our stochastic arrival model. To the best of our knowledge, such a result had not previously been established. 
    \item Optimality of Tail-Optimized LRU: if the user prompt length is deterministic, Expected-Tail-Optimized LRU is reduced to a deterministic version as stated in Policy \ref{alg:tail-optimized-lru} (with estimate $\hat{Q}$ replaced by deterministic $Q$). Theorem \ref{thm: optimality of tail LRU rand} thus establishes the optimality of Tail-Optimized LRU in this model.
    \item When $Q=0$ after the first turn and responses also have zero length, our model reduces to classic caching with unit page size. In this case, the optimal policy ``evicts the block whose conversation is least likely to return'', generalizing least-recently-used.
\end{itemize}
Therefore, Expected-Tail-Optimized LRU serves as a common backbone across three classic caching regimes, providing theoretical justification for adopting LRU and Tail-Optimized LRU for LLM inference workload.

\section{Experiments}\label{sec: experiment}

\subsection{Datasets and Metrics}  
We evaluate our caching policies on two conversational datasets:  WildChat and ShareGPT (Figure~\ref{fig: wildchat distribution} and Figure~\ref{fig: sharegpt distribution}). For each dataset, we select conversations based on their arrival timestamps and extract the first 1000--2000 turns across these conversations. Specifically, we sample conversations in chronological order (by first-turn timestamp), split each conversation into individual turns, and simulate their arrivals following the observed timestamps in the trace.

We measure latency metrics (median, P90, P95, P99) and service-level objective (SLO) attainment under different caching policies. Our default experimental configuration uses Vicuna-7B served on a single A100 GPU via vLLM with tensor parallelism disabled (TP = 1) and without mixed batching. We present the results for WildChat below. Due to space constraints, Results for ShareGPT, which exhibit qualitatively similar patterns, are provided in Appendix~\ref{apx: additional experiments}.
\begin{figure}[!htbp]
    \centering
    \includegraphics[width=\linewidth]{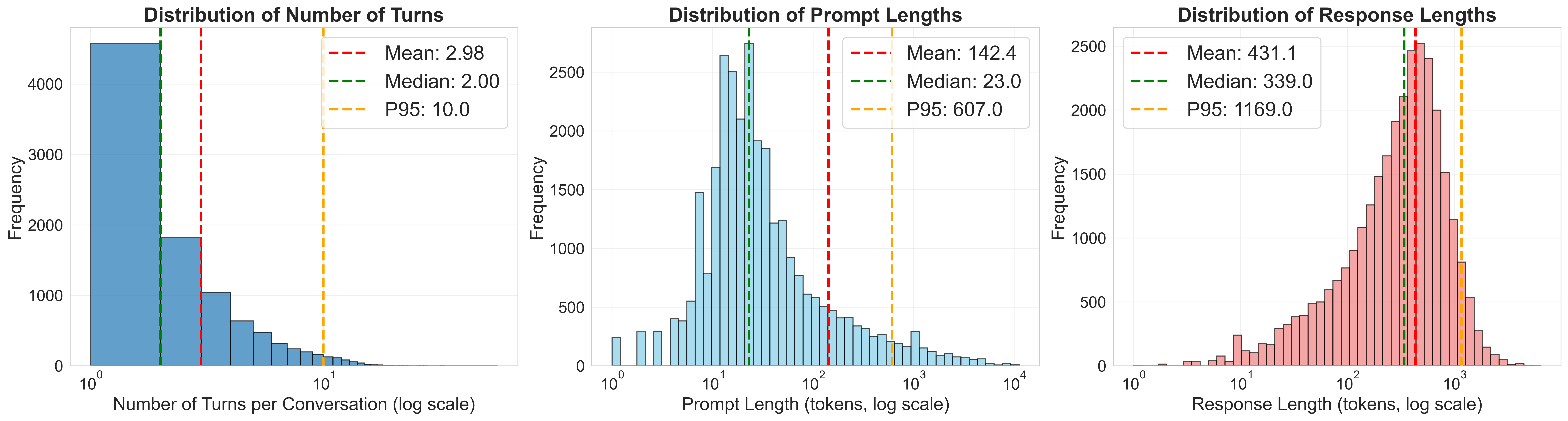}
    \caption{Distributions of turns and tokens of WildChat~\citep{zhao2024wildchat} datasets (sampled 10,000 conversations).}
    \label{fig: wildchat distribution}
\end{figure}

\subsection{Tail Latency Reduction}
We measure the tail latency reduction of our policy against LRU and \textit{Threshold LRU}---caches \emph{only if} the conversation length exceeds a fixed threshold, a configuration for LRU we observed in practice. We fix the input parameter for T-LRU, next-prompt length $\hat{Q}$, to be the average prompt length (200 for WildChat, 150 for ShareGPT), and use 1024 as the threshold for Threshold-LRU following the one used by OpenAI.\footnote{\url{https://platform.openai.com/docs/guides/prompt-caching}}

In Tables~\ref{tab:improvement-tlru-lru}--\ref{tab:improvement-tlru-thre}, we report the relative latency improvements of T-LRU over both LRU and Threshold-LRU across various cache capacities~$C$ and tail-latency thresholds~$\xi_{s}$.\footnote{$C = 10,000$ corresponds to approximately 4.8 GB for Vicuna-7B with Float16, see Appendix \ref{apx: KV cache size}.} We observe that T-LRU reduces P90 tail TTFT by up to 27.5\% and P95 tail TTFT by up to 23.9\% compared to LRU, and achieves comparable improvements over Threshold-LRU.

\noindent\textbf{Sensitivity to latency threshold $\xi_{s}$.}
The benefit of T-LRU peaks when latency threshold $\xi_{s}$ is calibrated to match the target tail percentile of interest.
\textit{For example, with capacity $C = 1000$ under LRU, medium, P90, P95, and P99 tail latencies are roughly 40 ms, 240 ms, 326 ms, and 505 ms.}
Setting $\xi_{s} = 200$ ms yields the largest improvement in P90 tail latency; setting $\xi_{s} = 300$ ms yields the largest improvement in P95 tail latency. 
A high value, e.g. $\xi_{s} = 500$ ms, relaxes protection for moderate tails and only provides protection for extreme tails ---up to 3\% of increase in P99 tail latency.\footnote{We set $\xi =\xi_s/c$ using a simple linear fit between awaited prefill tokens and prefill time, which is easy to implement in practice.}

Operationally, the latency threshold $\xi_{s}$ can be tuned either based on fixed service level objective (e.g., a target TTFT), or adaptively based on the observed tail latency of served turns. For example, the decision-maker can periodically update $\xi_{s}$ to match the desired tail latency TTFT observed over recent turns. 
Raising $\xi_{s}$ makes TEL-safe trimming more aggressive, and could potentially increase average latency as the policy allows more turns to incur latency up to $\xi_{s}$, reflecting the classic trade-off between average and tail performance.

\noindent\textbf{Comparing Threshold-LRU and T-LRU.}
Threshold-LRU is a straightforward patch for LRU’s blindness to conversation length. In industry systems such as at OpenAI, prompt caching is enabled only when the running conversation history exceeds a fixed length. While simple to implement, this policy makes a binary, all-or-nothing decision for each conversation: either cache the entire conversation history or cache nothing at all. In contrast, T-LRU makes adaptive, fine-grained caching decisions. Rather than caching entire conversations indiscriminately, T-LRU caches just enough tokens from each conversation to ensure its next request will not breach the latency threshold. This fundamental difference persists regardless of prompt length. Even with zero-length prompts, Threshold-LRU would apply its binary rule based on history length, while T-LRU would adaptively cache the amount that is just right so not impacting the tail.

\begin{table*}[!htbp]\centering
\caption{Relative latency improvement of T-LRU over LRU with various $\xi_s$ }
\scriptsize		
\begin{tabular}{rrrrrrrrrrrr}\toprule
&\multicolumn{2}{c}{$\xi_s$ = 50ms} &\multicolumn{2}{c}{$\xi_s$ = 100ms} &\multicolumn{2}{c}{$\xi_s$ = 150ms} &\multicolumn{2}{c}{$\xi_s$ = 200ms} &\multicolumn{3}{c}{$\xi_s$ = 500ms}  \\\cmidrule{1-12}
Capacity &p90 &p95 &p90 &p95 &p90 &p95 &p90 &p95 &p90 &p95 &p99  \\\midrule
1000 &\cellcolor[HTML]{c9e9d9}4.0\% &\cellcolor[HTML]{ffffff}0.0\% &\cellcolor[HTML]{c9e9d9}4.5\% &\cellcolor[HTML]{dcf1e6}1.0\% &\cellcolor[HTML]{b8e2ce}5.3\% &\cellcolor[HTML]{dcf1e6}1.5\% &\cellcolor[HTML]{b8e2ce}7.3\% &\cellcolor[HTML]{c9e9d9}2.0\% &\cellcolor[HTML]{ffd6d6}-1.4\% &\cellcolor[HTML]{ffd6d6}-1.3\% &\cellcolor[HTML]{c9e9d9}3.2\%  \\
2000 &\cellcolor[HTML]{dcf1e6}1.2\% &\cellcolor[HTML]{dcf1e6}0.6\% &\cellcolor[HTML]{c9e9d9}3.4\% &\cellcolor[HTML]{dcf1e6}0.6\% &\cellcolor[HTML]{b8e2ce}5.1\% &\cellcolor[HTML]{c9e9d9}3.4\% &\cellcolor[HTML]{b8e2ce}9.2\% &\cellcolor[HTML]{c9e9d9}4.4\% &\cellcolor[HTML]{ffd6d6}-4.9\% &\cellcolor[HTML]{ffd6d6}-2.8\% &\cellcolor[HTML]{c9e9d9}3.3\%  \\
4000 &\cellcolor[HTML]{dcf1e6}1.7\% &\cellcolor[HTML]{dcf1e6}0.7\% &\cellcolor[HTML]{c9e9d9}4.1\% &\cellcolor[HTML]{c9e9d9}2.8\% &\cellcolor[HTML]{96d4b6}10.5\% &\cellcolor[HTML]{c9e9d9}4.2\% &\cellcolor[HTML]{96d4b6}13.3\% &\cellcolor[HTML]{96d4b6}10.8\% &\cellcolor[HTML]{ffd6d6}-7.5\% &\cellcolor[HTML]{ffd6d6}-3.4\% &\cellcolor[HTML]{c9e9d9}3.4\%  \\
6000 &\cellcolor[HTML]{c9e9d9}5.0\% &\cellcolor[HTML]{c9e9d9}2.1\% &\cellcolor[HTML]{96d4b6}10.6\% &\cellcolor[HTML]{c9e9d9}4.0\% &\cellcolor[HTML]{85cdaa}16.0\% &\cellcolor[HTML]{96d4b6}11.4\% &\cellcolor[HTML]{96d4b6}14.3\% &\cellcolor[HTML]{85cdaa}15.4\% &\cellcolor[HTML]{ffd6d6}-8.7\% &\cellcolor[HTML]{ffd6d6}-4.2\% &\cellcolor[HTML]{c9e9d9}3.4\%  \\
8000 &\cellcolor[HTML]{c9e9d9}2.2\% &\cellcolor[HTML]{dcf1e6}0.7\% &\cellcolor[HTML]{96d4b6}10.5\% &\cellcolor[HTML]{b8e2ce}8.5\% &\cellcolor[HTML]{57bb8a}23.0\% &\cellcolor[HTML]{85cdaa}15.8\% &\cellcolor[HTML]{b8e2ce}8.8\% &\cellcolor[HTML]{85cdaa}19.5\% &\cellcolor[HTML]{ffd6d6}-13.7\% &\cellcolor[HTML]{ffd6d6}-3.5\% &\cellcolor[HTML]{dcf1e6}1.4\%  \\
10000 &\cellcolor[HTML]{c9e9d9}3.6\% &\cellcolor[HTML]{b8e2ce}7.4\% &\cellcolor[HTML]{85cdaa}20.0\% &\cellcolor[HTML]{85cdaa}15.7\% &\cellcolor[HTML]{57bb8a}27.5\% &\cellcolor[HTML]{57bb8a}20.1\% &\cellcolor[HTML]{b8e2ce}6.9\% &\cellcolor[HTML]{57bb8a}23.9\% &\cellcolor[HTML]{ffd6d6}-13.3\% &\cellcolor[HTML]{ffd6d6}-3.5\% &\cellcolor[HTML]{ffd6d6}-0.0\%  \\
\bottomrule
\end{tabular} \label{tab:improvement-tlru-lru}
\end{table*}
\begin{table*}[!htbp]\centering
\caption{Relative latency improvement of T-LRU over Threshold-LRU with various $\xi_s$}
\scriptsize	
\begin{tabular}{rrrrrrrrrrrr}\toprule
&\multicolumn{2}{c}{$\xi_s$ = 50ms} &\multicolumn{2}{c}{$\xi_s$ = 100ms} &\multicolumn{2}{c}{$\xi_s$ = 150ms} &\multicolumn{2}{c}{$\xi_s$ = 200ms} &\multicolumn{3}{c}{$\xi_s$ = 500ms}  \\\cmidrule{1-12}
Capacity &p90 &p95 &p90 &p95 &p90 &p95 &p90 &p95 &p90 &p95 &p99  \\\midrule
1000 &\cellcolor[HTML]{dcf1e6}1.5\% &\cellcolor[HTML]{ffffff}0.0\% &\cellcolor[HTML]{c9e9d9}2.0\% &\cellcolor[HTML]{dcf1e6}1.0\% &\cellcolor[HTML]{c9e9d9}2.9\% &\cellcolor[HTML]{dcf1e6}1.5\% &\cellcolor[HTML]{c9e9d9}4.8\% &\cellcolor[HTML]{c9e9d9}2.0\% &\cellcolor[HTML]{ffd6d6}-4.0\% &\cellcolor[HTML]{ffd6d6}-1.3\% &\cellcolor[HTML]{c9e9d9}3.2\%  \\
2000 &\cellcolor[HTML]{dcf1e6}0.4\% &\cellcolor[HTML]{dcf1e6}0.6\% &\cellcolor[HTML]{c9e9d9}2.7\% &\cellcolor[HTML]{dcf1e6}0.6\% &\cellcolor[HTML]{c9e9d9}4.4\% &\cellcolor[HTML]{c9e9d9}3.4\% &\cellcolor[HTML]{b8e2ce}8.6\% &\cellcolor[HTML]{c9e9d9}4.4\% &\cellcolor[HTML]{ffd6d6}-5.7\% &\cellcolor[HTML]{ffd6d6}-2.8\% &\cellcolor[HTML]{c9e9d9}3.3\%  \\
4000 &\cellcolor[HTML]{dcf1e6}0.3\% &\cellcolor[HTML]{ffffff}0.0\% &\cellcolor[HTML]{c9e9d9}2.7\% &\cellcolor[HTML]{c9e9d9}2.1\% &\cellcolor[HTML]{b8e2ce}9.3\% &\cellcolor[HTML]{c9e9d9}3.5\% &\cellcolor[HTML]{96d4b6}12.0\% &\cellcolor[HTML]{96d4b6}10.2\% &\cellcolor[HTML]{ffd6d6}-9.0\% &\cellcolor[HTML]{ffd6d6}-4.2\% &\cellcolor[HTML]{c9e9d9}3.4\%  \\
6000 &\cellcolor[HTML]{c9e9d9}4.7\% &\cellcolor[HTML]{dcf1e6}1.2\% &\cellcolor[HTML]{96d4b6}10.4\% &\cellcolor[HTML]{c9e9d9}3.1\% &\cellcolor[HTML]{85cdaa}15.7\% &\cellcolor[HTML]{96d4b6}10.6\% &\cellcolor[HTML]{96d4b6}14.1\% &\cellcolor[HTML]{96d4b6}14.6\% &\cellcolor[HTML]{ffd6d6}-9.1\% &\cellcolor[HTML]{ffd6d6}-5.2\% &\cellcolor[HTML]{c9e9d9}3.4\%  \\
8000 &\cellcolor[HTML]{dcf1e6}1.1\% &\cellcolor[HTML]{dcf1e6}0.7\% &\cellcolor[HTML]{b8e2ce}9.5\% &\cellcolor[HTML]{b8e2ce}8.5\% &\cellcolor[HTML]{57bb8a}22.2\% &\cellcolor[HTML]{85cdaa}15.8\% &\cellcolor[HTML]{b8e2ce}7.8\% &\cellcolor[HTML]{85cdaa}19.5\% &\cellcolor[HTML]{ffd6d6}-14.9\% &\cellcolor[HTML]{ffd6d6}-3.5\% &\cellcolor[HTML]{dcf1e6}1.4\%  \\
10000 &\cellcolor[HTML]{c9e9d9}2.4\% &\cellcolor[HTML]{b8e2ce}6.1\% &\cellcolor[HTML]{85cdaa}19.0\% &\cellcolor[HTML]{96d4b6}14.6\% &\cellcolor[HTML]{57bb8a}26.6\% &\cellcolor[HTML]{85cdaa}19.0\% &\cellcolor[HTML]{b8e2ce}5.7\% &\cellcolor[HTML]{57bb8a}22.8\% &\cellcolor[HTML]{ffd6d6}-14.7\% &\cellcolor[HTML]{ffd6d6}-5.0\% &\cellcolor[HTML]{ffd6d6}-0.0\%  \\
\bottomrule
\end{tabular} \label{tab:improvement-tlru-thre}
\end{table*}

\subsection{SLO Violation Reduction}
We now measure the \emph{count} of requests that a service-level objective, another objective similar to TEL---the \emph{amount} of latency beyond a threshold. Table \ref{tab:merged-improvement-tlru} reports the improvement on SLO attainment ratio (relative drop in requests whose TTFT exceeds 200 ms). With the latency threshold $\xi_{s}$ set to $200$ ms SLO, compared to LRU, T-LRU reduces between 8.8\% (small cache capacity) and 40.7\% (large cache capacity) of violations. 

\noindent\textbf{Sensitivity to latency threshold $\xi_{s}$.} 
When $\xi_{s}$ is much lower than the SLO ($\xi_{s} = 50$ or $100$ ms), the improvement is modest because both baselines already satisfy most requests; when $\xi_{s}$ is far higher ($\xi_{s} = 500$ ms), T-LRU focuses on larger tails and can allow for a few extra 200 ms violations.

These results echo the design goal: TEL minimization penalizes \emph{how much} a request overshoots $\xi_{s}$, yet the same trimming logic also cuts the \emph{number} of SLO violations whenever $\xi_{s}$ aligns or is slightly below the target latency budget.

\begin{table*}[!htbp]\centering
\caption{Relative improvement of T-LRU: \% reduction in requests with latency $>$ 200ms}
\scriptsize	
\begin{tabular}{rrrrrrrrrrr}\toprule
Capacity &\multicolumn{2}{c}{$\xi_s$ = 50ms} &\multicolumn{2}{c}{$\xi_s$ = 100ms} &\multicolumn{2}{c}{$\xi_s$ = 150ms} &\multicolumn{2}{c}{$\xi_s$ = 200ms} &\multicolumn{2}{c}{$\xi_s$ = 500ms} \\
 &LRU &Thre-LRU &LRU &Thre-LRU &LRU &Thre-LRU &LRU &Thre-LRU &LRU &Thre-LRU \\\cmidrule{1-11}
1000 & \cellcolor[HTML]{dcf1e6} 1.2\% & \cellcolor[HTML]{ffd6d6} -1.2\% & \cellcolor[HTML]{dcf1e6} 4.1\% & \cellcolor[HTML]{dcf1e6} 1.8\% & \cellcolor[HTML]{c9e9d9} 6.5\% & \cellcolor[HTML]{dcf1e6} 4.2\% & \cellcolor[HTML]{c9e9d9} 8.8\% & \cellcolor[HTML]{c9e9d9} 6.6\% & \cellcolor[HTML]{ffd6d6} -1.8\% & \cellcolor[HTML]{ffd6d6} -4.2\% \\
2000 & \cellcolor[HTML]{dcf1e6} 2.4\% & \cellcolor[HTML]{dcf1e6} 1.2\% & \cellcolor[HTML]{c9e9d9} 6.1\% & \cellcolor[HTML]{dcf1e6} 4.9\% & \cellcolor[HTML]{c9e9d9} 9.1\% & \cellcolor[HTML]{c9e9d9} 8.0\% & \cellcolor[HTML]{b8e2ce} 13.3\% & \cellcolor[HTML]{b8e2ce} 12.3\% & \cellcolor[HTML]{ffd6d6} -4.8\% & \cellcolor[HTML]{ffd6d6} -6.1\% \\
4000 & \cellcolor[HTML]{dcf1e6} 3.9\% & \cellcolor[HTML]{dcf1e6} 1.3\% & \cellcolor[HTML]{c9e9d9} 7.1\% & \cellcolor[HTML]{dcf1e6} 4.7\% & \cellcolor[HTML]{b8e2ce} 14.3\% & \cellcolor[HTML]{b8e2ce} 12.0\% & \cellcolor[HTML]{96d4b6} 22.1\% & \cellcolor[HTML]{b8e2ce} 20.0\% & \cellcolor[HTML]{ffd6d6} -12.3\% & \cellcolor[HTML]{ffd6d6} -15.3\% \\
6000 & \cellcolor[HTML]{dcf1e6} 2.0\% & \cellcolor[HTML]{dcf1e6} 0.7\% & \cellcolor[HTML]{b8e2ce} 12.2\% & \cellcolor[HTML]{b8e2ce} 11.0\% & \cellcolor[HTML]{96d4b6} 20.9\% & \cellcolor[HTML]{b8e2ce} 19.9\% & \cellcolor[HTML]{85cdaa} 30.4\% & \cellcolor[HTML]{96d4b6} 29.5\% & \cellcolor[HTML]{ffd6d6} -12.2\% & \cellcolor[HTML]{ffd6d6} -13.7\% \\
8000 & \cellcolor[HTML]{dcf1e6} 4.3\% & \cellcolor[HTML]{dcf1e6} 2.2\% & \cellcolor[HTML]{b8e2ce} 16.4\% & \cellcolor[HTML]{b8e2ce} 14.6\% & \cellcolor[HTML]{96d4b6} 29.3\% & \cellcolor[HTML]{96d4b6} 27.7\% & \cellcolor[HTML]{85cdaa} 33.6\% & \cellcolor[HTML]{85cdaa} 32.1\% & \cellcolor[HTML]{ffd6d6} -19.3\% & \cellcolor[HTML]{ffd6d6} -21.9\% \\
10000 & \cellcolor[HTML]{c9e9d9} 7.4\% & \cellcolor[HTML]{dcf1e6} 4.6\% & \cellcolor[HTML]{96d4b6} 25.9\% & \cellcolor[HTML]{96d4b6} 23.7\% & \cellcolor[HTML]{85cdaa} 31.9\% & \cellcolor[HTML]{96d4b6} 29.8\% & \cellcolor[HTML]{57bb8a} 40.7\% & \cellcolor[HTML]{85cdaa} 38.9\% & \cellcolor[HTML]{ffd6d6} -19.3\% & \cellcolor[HTML]{ffd6d6} -22.9\% \\
\bottomrule
\end{tabular} \label{tab:merged-improvement-tlru}
\end{table*}

\subsection{Offline Baselines and T-LRU Variants with Future Knowledge.}
Recent work on queueing systems has shown that machine learning predictions, such as predicted service times or job characteristics, can significantly improve scheduling and resource allocation decisions~\citep{mitzenmacher2025queueing}. In the context of LLM inference systems, various types of predictions (e.g., conversation continuation, prompt lengths) could also potentially enhance caching performance. This motivates a natural question: \textit{which predictions are most valuable for caching, and how much improvement can each type unlock?} To answer this, we design T-LRU variants with different levels of future knowledge, creating a predictability spectrum that quantifies the marginal benefit of each prediction type.

We evaluate three policies with increasing levels of future knowledge:
\begin{itemize}
\item \textbf{T-LRU (baseline):} Uses the empirical average prompt length to predict future requests. No knowledge of conversation termination or actual prompt lengths.
\item \textbf{End-Aware T-LRU:} Knows whether each conversation will continue (return) or terminate, but does not know future prompt lengths. This variant evicts all blocks from terminating conversations and follows standard T-LRU for continuing ones.
\item \textbf{Length-Aware T-LRU:} Knows both conversation continuations AND the exact length of the next user prompt. This variant uses exact prompt lengths for TEL-safe trimming and evicts all caches when conversations end.
\end{itemize}
Additionally, we include {Tail-Optimized Belady}, the hindsight-optimal policy for TEL that knows the entire future arrival sequence. This serves as an upper bound on achievable performance.

Figure~\ref{fig:result-latency-Wildchat} reports the latency distributions (median, P90, P95, P99) under different caching policies across varying cache capacities on the WildChat dataset. 

We highlight two observations from Figure \ref{fig:result-latency-Wildchat}. First, tail improvement is achieved with a modest cost to the median latency. T-LRU (blue squares) consistently beats LRU and Threshold-LRU at the tail (around 300 ms), but its median latency is slightly higher (increased from 10ms to 40ms). The trade-off is expected and in fact intended. 
Median latency degradation is minimal and users will likely not notice: for context, the duration of a blink is on average 100–400 milliseconds according to the Harvard Database of Useful Biological Numbers~\citep{bionumbers100706}. T-LRU deliberately trades off slightly higher median latency (a few milliseconds) to achieve larger tail latency reductions (tens to hundreds of milliseconds), which is more favorable for user experience.

Second, a single-bit forecast ``will this conversation continue?" is a remarkably powerful signal. End-Aware T-LRU performs much better than T-LRU, while Length-Aware T-LRU gains only a small additional edge from knowing the exact prompt length.
In practice, predicting whether a single conversation will continue is much easier than forecasting exact prompt sizes, and vastly easier than predicting the full arrival sequences required by Tail-Optimized Belady. 
Existing works like~\citep{jin2023s} propose models to predict the length of model response with up to 98.61\% prediction accuracy, underscoring the practical viability of deploying End-Aware policies.

Note that Tail-Optimized Belady is the optimal policy for our TEL objective for the $\xi_{s}$ chosen, thus it is not necessarily the optimal policy for tail latency at different levels. On the other hand, solving for the optimal policy for tail latency is computationally hard.

\begin{figure}[ht]
    \centering
    \begin{subfigure}[t]{\textwidth} 
        \centering \includegraphics[width=\linewidth]{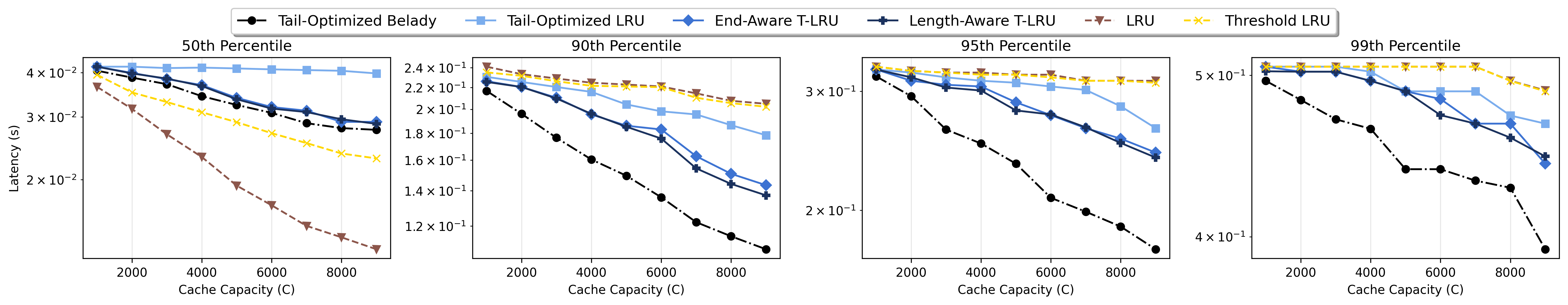}
    \end{subfigure} \hfill
    \begin{subfigure}[t]{\textwidth} 
        \centering \includegraphics[width=\linewidth]{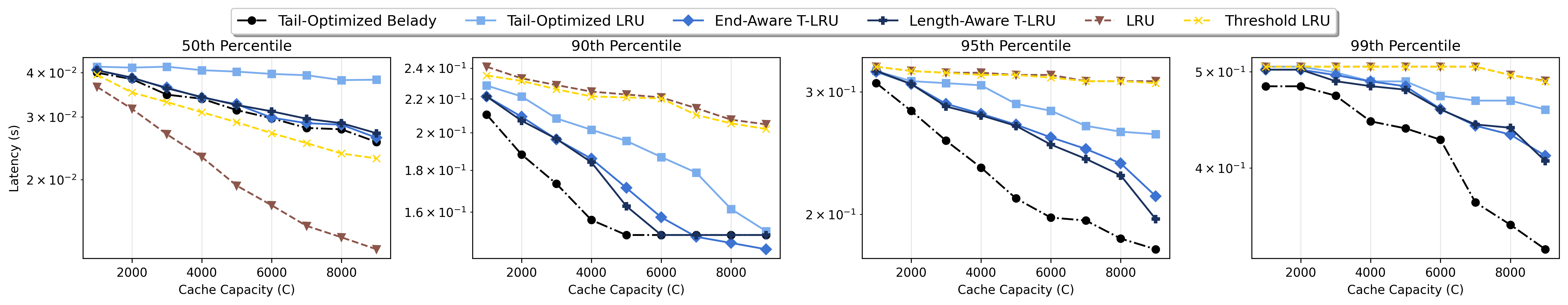}
    \end{subfigure}    
    \hfill
    \begin{subfigure}[t]{\textwidth} 
        \centering \includegraphics[width=\linewidth]{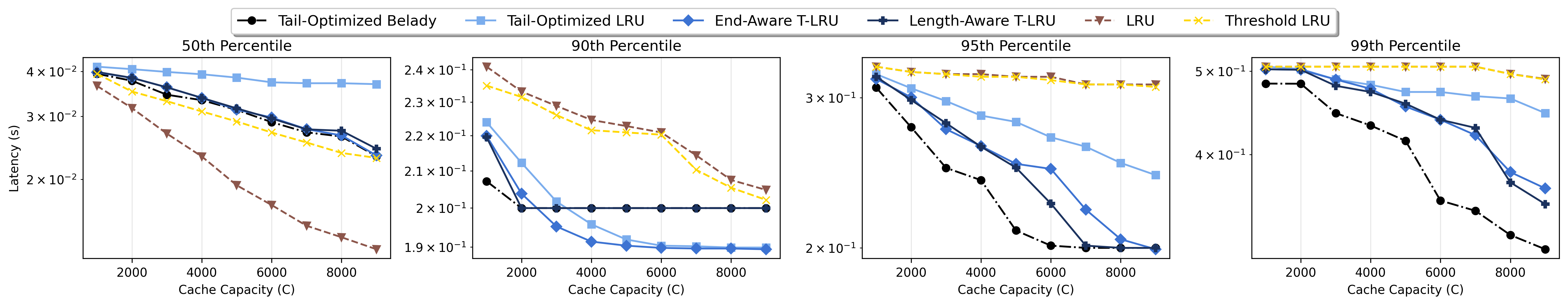}
    \end{subfigure}    
    \caption{Latency results for various settings (threshold latency $\xi_{s}$ = 100, 200, 300 ms) from top to bottom panels.}
    \label{fig:result-latency-Wildchat}
\end{figure}

\section{Future Directions and Conclusions}\label{sec:conclusion}

We propose Tail-Optimized LRU (T-LRU), a simple yet effective modification to the Least Recently Used caching policy that significantly improves tail latency in multi-turn conversational LLM serving. By adaptively caching just enough tokens to keep each conversation's next request below a latency threshold—rather than making binary all-or-nothing caching decisions, T-LRU achieves substantial tail latency reductions (e.g., 20--30\% improvement in P95 TTFT) with modest impact on median performance.

T-LRU makes a deliberate design choice: sacrifice tens of milliseconds at the median to eliminate hundreds of milliseconds at the tail. This trade-off aligns well with strict SLO requirements in production systems, though exploring multi-objective policies that balance tail and median latency differently remains an interesting direction for future work.

Our focus is on a single storage layer, but real systems increasingly adopt hierarchical memory architectures. Extending T-LRU to multi-tier KV caching systems—where evicted blocks migrate to slower storage (CPU memory, SSD) rather than being discarded—could unlock further performance gains. Recent work~\citep{qin2025mooncake} has begun exploring such architectures, and understanding optimal promotion/demotion policies across cache tiers represents an important open problem.
Another compelling direction is joint optimization of caching and load balancing. Recent work~\citep{srivatsa2024preble} has initiated exploration of how caching decisions interact with request routing in distributed LLM serving. Analyzing these problems through a queueing-theoretic lens and considering both cache locality and load distribution could systematically characterize fundamental trade-offs and guide practical system design.

\medskip

{
\small
\bibliographystyle{plainnat}
\bibliography{reference}

\begin{thebibliography}{43}
\providecommand{\natexlab}[1]{#1}
\providecommand{\url}[1]{\texttt{#1}}
\expandafter\ifx\csname urlstyle\endcsname\relax
  \providecommand{\doi}[1]{doi: #1}\else
  \providecommand{\doi}{doi: \begingroup \urlstyle{rm}\Url}\fi

\bibitem[Agrawal et~al.(2024)Agrawal, Kedia, Panwar, Mohan, Kwatra, Gulavani, Tumanov, and Ramjee]{agrawal2024taming}
Amey Agrawal, Nitin Kedia, Ashish Panwar, Jayashree Mohan, Nipun Kwatra, Bhargav Gulavani, Alexey Tumanov, and Ramachandran Ramjee.
\newblock Taming $\{$Throughput-Latency$\}$ tradeoff in $\{$LLM$\}$ inference with $\{$Sarathi-Serve$\}$.
\newblock In \emph{18th USENIX Symposium on Operating Systems Design and Implementation (OSDI 24)}, pages 117--134, 2024.

\bibitem[Anthropic(2024)]{Anthropic2024PromptCaching}
Anthropic.
\newblock Prompt caching, 2024.
\newblock URL \url{https://docs.anthropic.com/en/docs/build-with-claude/prompt-caching}.

\bibitem[B{\"a}uerle and Ott(2011)]{bauerle2011markov}
Nicole B{\"a}uerle and Jonathan Ott.
\newblock Markov decision processes with average-value-at-risk criteria.
\newblock \emph{Mathematical Methods of Operations Research}, 74:\penalty0 361--379, 2011.

\bibitem[Belady(1966)]{belady1966study}
Laszlo~A. Belady.
\newblock A study of replacement algorithms for a virtual-storage computer.
\newblock \emph{IBM Systems journal}, 5\penalty0 (2):\penalty0 78--101, 1966.

\bibitem[Berger et~al.(2017)Berger, Sitaraman, and Harchol-Balter]{berger2017adaptsize}
Daniel~S Berger, Ramesh~K Sitaraman, and Mor Harchol-Balter.
\newblock $\{$AdaptSize$\}$: Orchestrating the hot object memory cache in a content delivery network.
\newblock In \emph{14th USENIX Symposium on Networked Systems Design and Implementation (NSDI 17)}, pages 483--498, 2017.

\bibitem[Berger et~al.(2018)Berger, Berg, Zhu, Sen, and Harchol-Balter]{berger2018robinhood}
Daniel~S Berger, Benjamin Berg, Timothy Zhu, Siddhartha Sen, and Mor Harchol-Balter.
\newblock $\{$RobinHood$\}$: Tail latency aware caching--dynamic reallocation from $\{$Cache-Rich$\}$ to $\{$Cache-Poor$\}$.
\newblock In \emph{13th USENIX Symposium on Operating Systems Design and Implementation (OSDI 18)}, pages 195--212, 2018.

\bibitem[{BioNumbers Database}(2024)]{bionumbers100706}
{BioNumbers Database}.
\newblock Average duration of a single eye blink.
\newblock BioNumbers Database, 2024.
\newblock URL \url{https://bionumbers.hms.harvard.edu/bionumber.aspx?id=100706}.
\newblock BNID 100706.

\bibitem[Borodin and El-Yaniv(2005)]{borodin2005online}
Allan Borodin and Ran El-Yaniv.
\newblock \emph{Online computation and competitive analysis}.
\newblock cambridge university press, 2005.

\bibitem[Borodin et~al.(1995)Borodin, Irani, Raghavan, and Schieber]{borodin1995competitive}
Allan Borodin, Sandy Irani, Prabhakar Raghavan, and Baruch Schieber.
\newblock Competitive paging with locality of reference.
\newblock \emph{Journal of Computer and System Sciences}, 50\penalty0 (2):\penalty0 244--258, 1995.

\bibitem[Che et~al.(2002)Che, Tung, and Wang]{che2002hierarchical}
Hao Che, Ye~Tung, and Zhijun Wang.
\newblock Hierarchical web caching systems: Modeling, design and experimental results.
\newblock \emph{IEEE journal on Selected Areas in Communications}, 20\penalty0 (7):\penalty0 1305--1314, 2002.

\bibitem[Chen et~al.(2023)Chen, Sharma, Khan, Liu, Chang, Akella, Shakkottai, and Sitaraman]{chen2023darwin}
Jiayi Chen, Nihal Sharma, Tarannum Khan, Shu Liu, Brian Chang, Aditya Akella, Sanjay Shakkottai, and Ramesh~K Sitaraman.
\newblock Darwin: Flexible learning-based cdn caching.
\newblock In \emph{Proceedings of the ACM SIGCOMM 2023 Conference}, pages 981--999, 2023.

\bibitem[Chow et~al.(2015)Chow, Tamar, Mannor, and Pavone]{chow2015risk}
Yinlam Chow, Aviv Tamar, Shie Mannor, and Marco Pavone.
\newblock Risk-sensitive and robust decision-making: a cvar optimization approach.
\newblock \emph{Advances in neural information processing systems}, 28, 2015.

\bibitem[Chrobak and Noga(1999)]{chrobak1999lru}
Marek Chrobak and John Noga.
\newblock Lru is better than fifo.
\newblock \emph{Algorithmica}, 23:\penalty0 180--185, 1999.

\bibitem[Coffman and Denning(1973)]{coffman1973operating}
Edward~Grady Coffman and Peter~J Denning.
\newblock \emph{Operating systems theory}, volume 973.
\newblock prentice-Hall Englewood Cliffs, NJ, 1973.

\bibitem[Contributors(2025)]{sharegpt}
ShareGPT Contributors.
\newblock Sharegpt: A dataset of multi-turn chat interactions with large language models.
\newblock \url{https://huggingface.co/datasets/RyokoAI/ShareGPT52K}, 2025.

\bibitem[Dan and Towsley(1990)]{dan1990approximate}
Asit Dan and Don Towsley.
\newblock An approximate analysis of the lru and fifo buffer replacement schemes.
\newblock In \emph{Proceedings of the 1990 ACM SIGMETRICS conference on Measurement and modeling of computer systems}, pages 143--152, 1990.

\bibitem[Fiat et~al.(1991)Fiat, Karp, Luby, McGeoch, Sleator, and Young]{fiat1991competitive}
Amos Fiat, Richard~M Karp, Michael Luby, Lyle~A McGeoch, Daniel~D Sleator, and Neal~E Young.
\newblock Competitive paging algorithms.
\newblock \emph{Journal of Algorithms}, 12\penalty0 (4):\penalty0 685--699, 1991.

\bibitem[Gim et~al.(2024)Gim, Chen, Lee, Sarda, Khandelwal, and Zhong]{gim2024prompt}
In~Gim, Guojun Chen, Seung-seob Lee, Nikhil Sarda, Anurag Khandelwal, and Lin Zhong.
\newblock Prompt cache: Modular attention reuse for low-latency inference.
\newblock \emph{Proceedings of Machine Learning and Systems}, 6:\penalty0 325--338, 2024.

\bibitem[Horton et~al.(2024)Horton, Cao, Sun, Jin, Mehta, Rastegari, and Nabi]{horton2024kv}
Maxwell Horton, Qingqing Cao, Chenfan Sun, Yanzi Jin, Sachin Mehta, Mohammad Rastegari, and Moin Nabi.
\newblock Kv prediction for improved time to first token.
\newblock \emph{arXiv preprint arXiv:2410.08391}, 2024.

\bibitem[Jain(1990)]{jain1990characteristics}
Raj Jain.
\newblock Characteristics of destination address locality in computer networks: A comparison of caching schemes.
\newblock \emph{Computer networks and ISDN systems}, 18\penalty0 (4):\penalty0 243--254, 1990.

\bibitem[Jin et~al.(2024)Jin, Zhang, Jiang, Liu, Liu, Liu, and Jin]{jin2024ragcache}
Chao Jin, Zili Zhang, Xuanlin Jiang, Fangyue Liu, Xin Liu, Xuanzhe Liu, and Xin Jin.
\newblock Ragcache: Efficient knowledge caching for retrieval-augmented generation.
\newblock \emph{arXiv preprint arXiv:2404.12457}, 2024.

\bibitem[Jin et~al.(2023)Jin, Wu, Brooks, and Wei]{jin2023s}
Yunho Jin, Chun-Feng Wu, David Brooks, and Gu-Yeon Wei.
\newblock $s^3$: Increasing gpu utilization during generative inference for higher throughput.
\newblock \emph{Advances in Neural Information Processing Systems}, 36:\penalty0 18015--18027, 2023.

\bibitem[Kamath et~al.(2025)Kamath, Prabhu, Mohan, Peter, Ramjee, and Panwar]{kamath2025pod}
Aditya~K Kamath, Ramya Prabhu, Jayashree Mohan, Simon Peter, Ramachandran Ramjee, and Ashish Panwar.
\newblock Pod-attention: Unlocking full prefill-decode overlap for faster llm inference.
\newblock In \emph{Proceedings of the 30th ACM International Conference on Architectural Support for Programming Languages and Operating Systems, Volume 2}, pages 897--912, 2025.

\bibitem[King(1972)]{king1972analysis}
WC~King.
\newblock Analysis of paging algorithms.
\newblock In \emph{Proc. IFIP 1971 Congress, Ljubljana}, pages 485--490. North-Holland, 1972.

\bibitem[Kwon et~al.(2023)Kwon, Li, Zhuang, Sheng, Zheng, Yu, Gonzalez, Zhang, and Stoica]{kwon2023efficient}
Woosuk Kwon, Zhuohan Li, Siyuan Zhuang, Ying Sheng, Lianmin Zheng, Cody~Hao Yu, Joseph Gonzalez, Hao Zhang, and Ion Stoica.
\newblock Efficient memory management for large language model serving with pagedattention.
\newblock In \emph{Proceedings of the 29th Symposium on Operating Systems Principles}, pages 611--626, 2023.

\bibitem[Leonardi and Torrisi(2015)]{leonardi2015least}
Emilio Leonardi and Giovanni~Luca Torrisi.
\newblock Least recently used caches under the shot noise model.
\newblock In \emph{2015 IEEE Conference on Computer Communications (INFOCOM)}, pages 2281--2289. IEEE, 2015.

\bibitem[Lykouris and Vassilvitskii(2021)]{lykouris2021competitive}
Thodoris Lykouris and Sergei Vassilvitskii.
\newblock Competitive caching with machine learned advice.
\newblock \emph{Journal of the ACM (JACM)}, 68\penalty0 (4):\penalty0 1--25, 2021.

\bibitem[Mhaisen et~al.(2022)Mhaisen, Sinha, Paschos, and Iosifidis]{mhaisen2022optimistic}
Naram Mhaisen, Abhishek Sinha, Georgios Paschos, and George Iosifidis.
\newblock Optimistic no-regret algorithms for discrete caching.
\newblock \emph{Proceedings of the ACM on Measurement and Analysis of Computing Systems}, 6\penalty0 (3):\penalty0 1--28, 2022.

\bibitem[Mitzenmacher and Shahout(2025)]{mitzenmacher2025queueing}
Michael Mitzenmacher and Rana Shahout.
\newblock Queueing, predictions, and large language models: Challenges and open problems.
\newblock \emph{Stochastic Systems}, 15\penalty0 (3):\penalty0 195--219, 2025.

\bibitem[{OpenAI}(2024)]{OpenAI2024PromptCaching}
{OpenAI}.
\newblock Prompt caching guide, 2024.
\newblock URL \url{https://platform.openai.com/docs/guides/prompt-caching}.
\newblock Developer documentation.

\bibitem[{OpenAI Newsroom}(2024)]{OpenAI2024Stats}
{OpenAI Newsroom}.
\newblock 300m weekly chatgpt users and 1b daily messages, December 2024.
\newblock URL \url{https://x.com/OpenAINewsroom/status/1864373399218475440}.
\newblock Tweet.

\bibitem[Phalke and Gopinath(1995)]{phalke1995inter}
Vidyadhar Phalke and Bhaskarpillai Gopinath.
\newblock An inter-reference gap model for temporal locality in program behavior.
\newblock \emph{ACM SIGMETRICS Performance Evaluation Review}, 23\penalty0 (1):\penalty0 291--300, 1995.

\bibitem[Qin et~al.(2025)Qin, Li, He, Cui, Ren, Zhang, Wu, Zheng, and Xu]{qin2025mooncake}
Ruoyu Qin, Zheming Li, Weiran He, Jialei Cui, Feng Ren, Mingxing Zhang, Yongwei Wu, Weimin Zheng, and Xinran Xu.
\newblock Mooncake: Trading more storage for less computation—a kvcache-centric architecture for serving llm chatbot.
\newblock In \emph{23rd USENIX Conference on File and Storage Technologies (FAST 25)}, pages 155--170. USENIX Association, 2025.

\bibitem[Sleator and Tarjan(1985)]{sleator1985amortized}
Daniel~D Sleator and Robert~E Tarjan.
\newblock Amortized efficiency of list update and paging rules.
\newblock \emph{Communications of the ACM}, 28\penalty0 (2):\penalty0 202--208, 1985.

\bibitem[Srivatsa et~al.(2024)Srivatsa, He, Abhyankar, Li, and Zhang]{srivatsa2024preble}
Vikranth Srivatsa, Zijian He, Reyna Abhyankar, Dongming Li, and Yiying Zhang.
\newblock Preble: Efficient distributed prompt scheduling for llm serving.
\newblock 2024.

\bibitem[{vLLM Team}(2025)]{vLLM2025PrefixCaching}
{vLLM Team}.
\newblock Automatic prefix caching, 2025.
\newblock URL \url{https://docs.vllm.ai/en/stable/automatic_prefix_caching/details.html}.
\newblock Project documentation.

\bibitem[Ye et~al.(2025)Ye, Chen, Lai, Lin, Zhang, Wang, Chen, Kasikci, Grover, Krishnamurthy, et~al.]{ye2025flashinfer}
Zihao Ye, Lequn Chen, Ruihang Lai, Wuwei Lin, Yineng Zhang, Stephanie Wang, Tianqi Chen, Baris Kasikci, Vinod Grover, Arvind Krishnamurthy, et~al.
\newblock Flashinfer: Efficient and customizable attention engine for llm inference serving.
\newblock \emph{arXiv preprint arXiv:2501.01005}, 2025.

\bibitem[Zhao et~al.(2024)Zhao, Ren, Hessel, Cardie, Choi, and Deng]{zhao2024wildchat}
Wenting Zhao, Xiang Ren, Jack Hessel, Claire Cardie, Yejin Choi, and Yuntian Deng.
\newblock Wildchat: 1m chatgpt interaction logs in the wild.
\newblock \emph{arXiv preprint arXiv:2405.01470}, 2024.

\bibitem[Zheng et~al.(2024)Zheng, Yin, Xie, Sun, Huang, Yu, Cao, Kozyrakis, Stoica, Gonzalez, et~al.]{zheng2024sglang}
Lianmin Zheng, Liangsheng Yin, Zhiqiang Xie, Chuyue~Livia Sun, Jeff Huang, Cody~Hao Yu, Shiyi Cao, Christos Kozyrakis, Ion Stoica, Joseph~E Gonzalez, et~al.
\newblock Sglang: Efficient execution of structured language model programs.
\newblock \emph{Advances in Neural Information Processing Systems}, 37:\penalty0 62557--62583, 2024.

\bibitem[Zhong et~al.(2024)Zhong, Liu, Chen, Hu, Zhu, Liu, Jin, and Zhang]{zhong2024distserve}
Yinmin Zhong, Shengyu Liu, Junda Chen, Jianbo Hu, Yibo Zhu, Xuanzhe Liu, Xin Jin, and Hao Zhang.
\newblock $\{$DistServe$\}$: Disaggregating prefill and decoding for goodput-optimized large language model serving.
\newblock In \emph{18th USENIX Symposium on Operating Systems Design and Implementation (OSDI 24)}, pages 193--210, 2024.

\bibitem[Zhu et~al.(2023)Zhu, Sheng, Zheng, Barrett, Jordan, and Jiao]{zhu2023optimal}
Banghua Zhu, Ying Sheng, Lianmin Zheng, Clark Barrett, Michael~I Jordan, and Jiantao Jiao.
\newblock On optimal caching and model multiplexing for large model inference.
\newblock \emph{arXiv preprint arXiv:2306.02003}, 2023.

\bibitem[Zhu et~al.(2024)Zhu, Zhao, Zhao, Zuo, Gu, Xie, Gao, Xu, Tang, Ye, et~al.]{zhu2024nanoflow}
Kan Zhu, Yilong Zhao, Liangyu Zhao, Gefei Zuo, Yile Gu, Dedong Xie, Yufei Gao, Qinyu Xu, Tian Tang, Zihao Ye, et~al.
\newblock Nanoflow: Towards optimal large language model serving throughput.
\newblock \emph{arXiv preprint arXiv:2408.12757}, 2024.

\bibitem[Zhu et~al.(2016)Zhu, Berger, and Harchol-Balter]{zhu2016snc}
Timothy Zhu, Daniel~S Berger, and Mor Harchol-Balter.
\newblock Snc-meister: Admitting more tenants with tail latency slos.
\newblock In \emph{Proceedings of the Seventh ACM Symposium on Cloud Computing}, pages 374--387, 2016.

\end{thebibliography}
}

\newpage
\appendix


\section*{Appendix}
\section{Proof of Theorem \ref{thm: modified baledy is optimal} Hindsight Optimal Policy for TEL}\label{apx: hindsight optimal policy}
\begin{proof}
The proof has two steps: 1) show that optimizing the original TEL problem \eqref{eq: opt for min TEL} is equivalent to a new problem \eqref{eq: opt for belady} of maximizing the total number of reused cached blocks, subject to an additional ``TEL-safe" capping constraint; 2) show that this new problem is a classic offline caching problem whose optimal solution is to cap the KV cache size using TEL-safe budget and then evict using furthest-in-future policy.

\paragraph{Step one: Equivalence of Optimization Problems}
Fix a feasible $\bm{x}$ to \eqref{eq: opt for min TEL}, the optimal $\bm{u}$ is given by
\[u_{i,t} = \left(\sum_{j=1}^{t} q_{i,j} + \sum_{j=1}^{t-1} a_{i,j} - x_{i,t} - \xi\right)^+,\]
as otherwise we have $u_{i,t}$ infeasible or can be improved. Thus we can focus on cache decisions $\bm{x}$. 

We claim the optimal solution to the optimization problem \eqref{eq: opt for min TEL} must satisfy the TEL-safe budget $x_{i,t} \leq \left(\sum_{j=1}^{t} q_{i,j} + \sum_{j=1}^{t-1} a_{i,j} -\xi\right)^+$ for every $i \in [N],t \in \mathcal{T}_i$. 
Suppose not, i.e., the optimal $\bm{x}'$ to \eqref{eq: opt for min TEL} satisfies $x_{i,t}' > \left(\sum_{j=1}^{t} q_{i,j} + \sum_{j=1}^{t-1} a_{i,j} -\xi\right)^+$ for some $i \in [N],t \in \mathcal{T}_i$. Then we have 
\[\sum_{j=1}^{t} q_{i,j} + \sum_{j=1}^{t-1} a_{i,j} - x'_{i,t}-\xi <0, u_{i,t}' = 0.\]
In this case, setting $x_{i,t} = \left(\sum_{j=1}^{t} q_{i,j} + \sum_{j=1}^{t-1} a_{i,j} -\xi\right)^+$ will not affect the value of $u_{i,t}'$, while releasing cache capacity that can be directed to other conversations, which contradicts the optimality of $\bm{x}'$.

As the optimal solution must respect the TEL-safe budget, we can simplify the objective: 
\begin{align*}
    \sum_{i\in [N]}\sum_{t\in \mathcal{T}_i} u_{i,t} & = \sum_{i \in [N]}\sum_{t\in \mathcal{T}_i}\left(\sum_{j=1}^{t} q_{i,j} + \sum_{j=1}^{t-1} a_{i,j} - x_{i,t}-\xi\right)^+\\
    & =  \sum_{i \in [N]}\sum_{t\in \mathcal{T}_i} \mathbbm{1}\left\{\sum_{j=1}^{t} q_{i,j} + \sum_{j=1}^{t-1} a_{i,j} -\xi \geq 0 \right\} \cdot \left(\sum_{j=1}^{t} q_{i,j} + \sum_{j=1}^{t-1} a_{i,j} - x_{i,t}- \xi\right) \\
    & \quad \quad + \sum_{i \in [N]}\sum_{t\in \mathcal{T}_i} \mathbbm{1}\left\{\sum_{j=1}^{t} q_{i,j} + \sum_{j=1}^{t-1} a_{i,j} -\xi < 0 \right\}\cdot 0\\
    & = \sum_{i \in [N]}\sum_{t\in \mathcal{T}_i} \mathbbm{1}\left\{\sum_{j=1}^{t} q_{i,j} + \sum_{j=1}^{t-1} a_{i,j} -\xi \geq 0 \right\}  \left(\sum_{j=1}^{t} q_{i,j} + \sum_{j=1}^{t-1} a_{i,j} - \xi\right) -\sum_{i \in [N]} \sum_{t\in \mathcal{T}_i} x_{i,t}
\end{align*}
where the last equality holds as $x_{i,t} =0$ for the requests with $\sum_{j=1}^{t} q_{i,j} + \sum_{j=1}^{t-1} a_{i,j} -\xi < 0$.
Thus minimizing $\sum_{i\in [N]}\sum_{t\in \mathcal{T}_i} u_{i,t}$ is equivalent to maximizing $\sum_{i \in [N]} \sum_{t\in \mathcal{T}_i} x_{i,t}$ as the first term of the right-hand-side of the last equality above is a constant.

Therefore, we can rewrite \eqref{eq: opt for min TEL} as follows:
\begin{align}
    \max_{x_{i,t} \in \mathbb{N}} \; & \sum_{i\in [N]}\sum_{t\in \mathcal{T}_i} x_{i,t} \label{eq: opt for belady}\\
    \text{s.t. } \; & \eqref{cons: capacity}, \eqref{cons: optional caching}, \eqref{cons: cache on arrival} \notag\\
    & x_{i,t} \leq \left(\sum_{j=1}^{t}  q_{i,j} + \sum_{j=1}^{t-1} a_{i,j} -\xi\right)^+, \forall i \in [N], t\in \mathcal{T}_i; \label{cons: cache just enough}
\end{align}

\paragraph{Step two.} 

We can interpret each block as a unit-size item that, once arrived, is requested again at every subsequent turn of the same conversation (prefix reuse). At time $t \in \mathcal{T}_i$, the number of cache hits for that request of conversation $i$ equals exactly $x_{i,t}$ Hence the objective of \eqref{eq: opt for belady} is the total number of cache hits over the horizon.

Because the cache content may change only at request times and all items have unit size, \eqref{eq: opt for belady} is a paging instance with unit pages and clairvoyant knowledge, except that each request carries a per-request cap on the number of allowable hits, which we can enforce by immediately discarding any excess above upon the turn’s completion.

In paging with unit pages, the offline optimal policy that maximizes total hits (equivalently minimizes misses) is Bélády’s furthest-in-future rule: whenever eviction is needed, evict the item whose next request is farthest in the future. Here, all blocks of conversation $i$ share the same next request time---the next arrival of conversation $i$---so the rule specializes to: evict from the conversation whose next arrival is farthest in the future. 

\end{proof}

\section{Proof of Theorem \ref{thm: optimality of tail LRU rand}}\label{apx: proof of ETLRU}
\begin{proof}
Given system state $\bm{\lambda}, \bm{L}, \bm{X}$, let $\theta$ denote the index of the conversation that is the $k^{th}$ arrival with user prompt length $Q$ and model response length $A$. The finite-horizon value-to-go function is
\begin{align*} 
    & V_k(\bm{\lambda}, \bm{L}, \bm{X}, \theta, Q, A) = \underbrace{(L_\theta + Q -X_\theta - \xi)^+}_{\textsf{number of uncached blocks above threshold}} \\
    & \quad +  \min_{\bm{X}' \in \mathcal{X}(\bm{X},\theta, L_\theta + Q+A)}\mathbb{E}_{ \tau, \theta', Q', A'}\left[V_{k+1}(\underbrace{\Phi(\bm{\lambda}, \theta, \tau)}_{\text{belief arrival state transition}}, \underbrace{\Psi(\bm{L}, \theta, Q+A)}_{\text{conversation length transition}},\bm{X}', \theta', Q', A')\right]
\end{align*}
with $V_{M+1}(\cdot)=0$, where the feasible caching decision space is
        \[\mathcal{X}(\bm{X}, \theta, L) = \{\bm{Y} \in \mathbb{N}^{\textsf{dim}(\bX, \theta)}: \sum_{i} Y_i \leq C, 0\leq Y_\theta \leq L, 0 \leq Y_i \leq X_i, i\neq \theta\}\]
with $\textsf{dim}(\bX, \theta) = \textsf{dim}(\bX) + \mathbbm{1}\{\theta > \textsf{dim}(\bX)\}$, here $\textsf{dim}(\bX)$ denotes the dimension of vector $\bX$, and the dimension expands when a new conversation arrives;  $\Phi(\blam, \theta, \tau)$ update the belief turn rate of conversation $\theta$ to $\bar{\lambda}_i$, then discount belief turn rates of all conversations by $\exp(-\mu \tau)$; $\Psi(\bL, \theta, Q+A)$ updates the conversation length vector. Specifically, we increase the dimension of $\bL$ if necessary (i.e., when $\theta$ represents a new conversation), and add $Q + A$ to its $\theta^{th}$ entry.  

Let $\theta$ denote the index of the conversation that is the $k^{th}$ arrival, with user prompt length $Q$ and model response length $A$. Define the belief arrival rate vector as $\blam$, the conversation length vector as $\bL$, and the cached token length vector as $\bX$, the cost-to-go function is given by:
\begin{align*} 
    & V_k(\bm{\lambda}, \bm{L}, \bm{X}, \theta, Q, A) = {(L_\theta + Q -X_\theta - \xi)^+} \\
    & \quad +  \min_{\bm{X}' \in \mathcal{X}(\bm{X},\theta, L_\theta + Q+A)}\mathbb{E}_{ \tau, \theta', Q', A'}\left[V_{k+1}({\Phi(\bm{\lambda}, \theta, \tau)}, {\Psi(\bm{L}, \theta, Q+A)},\bm{X}', \theta', Q', A')\right]
\end{align*}
with $V_{M+1}(\cdot)=0$. 
Let's rewrite $\Phi(\bm{\lambda}, \theta, \tau) = \Phi(\Gamma(\bm{\lambda}, \theta), \tau)$ with 
\begin{itemize}
    \item $\Gamma(\blam, \theta)$ updates the return rate vector upon the arrival of conversation $\theta$. Specifically, this operator changes the belief arrival rate of conversation $\theta$ to $\bar{\lambda}_i$.
    \item $\Phi(\blam, \tau)$ discount all return rates by $\exp(-\mu \tau)$. 
\end{itemize}

The proof is by using induction and argue that if we choose a different caching state than $\bm{X}^{\textsf{TLRU}}$, the cost will be higher.
At last arrival $M$, $V_M(\blam, \bL, \bm{X}, \theta, Q_\theta, A_\theta) = (L_\theta + Q_\theta - X_\theta - \xi)^+$ and any caching policy is optimal. 

Suppose this holds for the $k^{th}+1$ arrival. 
We proceed to show that the result holds for the $k^{th}$ arrival.
To simplify the notation, we define
\[J_k(\blam, \bL, \bX) = \mathbb{E}_{\tau, \theta, Q_\theta, A_\theta}[V_k(\Phi(\blam, \tau), \bL, \bX, \theta, Q_\theta, A_\theta)], \tilde{\blam} = \Gamma(\blam, \theta), \tilde{\bL} = \Psi(\bL, \theta, Q_\theta + A_\theta),\]
Then we need to prove
\[J_{k+1}(\tilde{\blam}, \tilde{\bL}, \bX^\textsf{ETLRU}) \leq J_{k+1}(\tilde{\blam}, \tilde{\bL}, \bX')\]

To see this, by definition of state transition, suppose the inter-arrival time is $\tau$, the discounted arrival rates are 
\[\tilde{\lambda}_i\cdot \exp(-\mu \tau) \text{ with } \tilde{\lambda}_\theta = \bar{\lambda}_{\theta}.\] 
From these expressions, we can conclude that regardless of value of $\tau$, the return rates at next arrival maintain the same relative ordering as in $\bm{\lambda}$ for conversations. Note that due to heterogeneous turn rates across conversations, conversation $\theta$ that arrived at the $k^{th}$ arrival may not have the highest turn rate. 
Without loss of generality, let's assume $\tau = 0$ and let $\tilde{p}_j = \tilde{\lambda}_j /(\lambda_{\textsf{conv}} + \sum_i \tilde{\lambda}_i)$ denote the probability that conversation $j$ returns at the next arrival, and $\tilde{p}_{\textsf{dim}(\bX(1))+1} =\lambda_{\textsf{conv}} /(\lambda_{\textsf{conv}} + \sum_i \tilde{\lambda}_i)$ denote the probability that a new conversation starts at the next arrival.

We proceed to show this holds for any number of tokens evicted by treating each token eviction separately. 
Among all conversations that have arrived so far and have at least one cached token, list their indices as $i(1), i(2), \ldots$ in ascending order of the ranking criterion score
\[\tilde{\lambda}_i \mathbb{P}(\tilde{L}_i + Q_i -\xi \geq X_i),\] 
Insert conversation $\theta$ that just arrives into this ordered list according to its own score 
\[\bar{\lambda}_\theta \mathbb{P}( Q_\theta -\xi \geq 0)\]

Let $\bX(1)$ denote the cache state that evicts a token from conversation $i(1)$, and $\bX(k)$ denote the cache state that evicts a token from conversation $i(k)$ with $k>1$. 
\begin{align*}
    & J_{k+1}(\tilde{\blam}, \tilde{L}, \bX(1)) \\
    & = \sum_{i\in \textsf{dim}(\bX(1))+1}\tilde{p}_i \mathbb{E}_{Q_i}[(\tilde{L}_i + Q_i - X_i(1) - \xi)^+] \\
    & \quad \quad + \tilde{p}_{i(1)} \mathbb{E}_{Q_{i(1)},A_{i(1)}} \left[\min_{\bX'(1) \in \mathcal{X}(\bX(1), i(1), \tilde{L}_{i(1)}+Q_{i(1)} + A_{i(1)}) } J_{k+2}(\Gamma(\tilde{\blam}, i(1)), \Psi(\tilde{L}, i(1), Q_{i(1)}+ A_{i(1)}), \bX'(1))\right] \\
    & \quad \quad + \tilde{p}_{i(k)} \mathbb{E}_{Q_{i(k)}, A_{i(k)}} \left[\min_{\bX'(1) \in \mathcal{X}(\bX(1), i(k), \tilde{L}_{i(k)}+Q_{i(k)}+A_{i(k)}) } J_{k+2}(\Gamma(\tilde{\blam}, i(k)), \Psi(\tilde{L}, i(k),Q_{i(k)}+A_{i(k)}), \bX'(1))\right] \\
    & \quad \quad + \sum_{i\neq i(1), i(k)} \mathbb{E}_{Q_i, A_i}\left[\min_{\bX'(1) \in \mathcal{X}(\bX(1), i, \tilde{L}_i+Q_i+A_i) } J_{k+2}(\Gamma(\tilde{\blam}, i), \Psi(\tilde{L}, i, Q_i + A_i), \bX'(1))\right]\\
    & \leq \sum_{i\in \textsf{dim}(\bX(1))+1}\tilde{p}_i \mathbb{E}_{Q_i}[(\tilde{L}_i + Q_i - X_i(k) - \xi)^+] \\
    & \quad \quad + \tilde{p}_{i(1)} \mathbb{E}_{Q_{i(1)}, A_{i(1)}} \left[\min_{\bX'(k) \in \mathcal{X}(\bX(k), i(1), \tilde{L}_{i(1)}+Q_{i(1)}+A_{i(1)}) } J_{k+2}(\Gamma(\tilde{\blam}, i(1)), \Psi(\tilde{L}, i(1), Q_{i(1)}+A_{i(1)}), \bX'(k))\right] \\
    & \quad \quad + \tilde{p}_{i(k)} \mathbb{E}_{Q_{i(k)}, A_{i(k)}} \left[\min_{\bX'(k) \in \mathcal{X}(\bX(k), i(k), \tilde{L}_{i(k)}+Q_{i(k)} + A_{i(k)}) } J_{k+2}(\Gamma(\tilde{\blam}, i(k)), \Psi(\tilde{L}, i(k),  Q_{i(k)}+A_{i(k)}), \bX'(k))\right] \\
    & \quad \quad + \sum_{i\neq i(1), i(k)} \mathbb{E}_{Q_i, A_i}\left[\min_{\bX'(k) \in \mathcal{X}(\bX(k), i, \tilde{L}_i+Q_i+A_i) } J_{k+2}(\Gamma(\tilde{\blam}, i), \Psi(\tilde{L}, i, Q_i + A_i), \bX'(k))\right]\\
    & =  J_{k+1}(\tilde{\blam}, \tilde{L}, \bX(k)),
\end{align*}
where the inequality holds as 
\begin{itemize}
    \item by Definition \ref{def: TLRU rand}, $\bX(1)$ is the optimal solution while $\bX(k)$ is a feasible solution, and $\tilde{p}_i$ are proportional to $\tilde{\lambda}_i$, thus
        \[\sum_{i \in \textsf{dim}(\bX(1))}\tilde{p}_i  \mathbb{E}_{Q_i}[(\tilde{L}_i + Q - X_i(1) - \xi)^+] \leq \sum_{i \in \textsf{dim}(\bX(1))}\tilde{p}_i\mathbb{E}_{Q_i}[(\tilde{L}_i + Q - X_i(k) - \xi)^+]\]
    and the expected cost incurred when a new conversation arrives (with probability $\tilde{p}_{\textsf{dim}(\bX(1))+1}$) is $\mathbb{E}_{Q}[(Q-\xi)^+]$ for both caching state, thus 
    \[\sum_{i \in \textsf{dim}(\bX(1))+1}\tilde{p}_i \mathbb{E}_{Q_i}[(\tilde{L}_i + Q_i - X_i(1) - \xi)^+] \leq \sum_{i\in \textsf{dim}(\bX(1))+1}\tilde{p}_i \mathbb{E}_{Q_i}[(\tilde{L}_i + Q_i - X_i(k) - \xi)^+]\]
    \item by induction hypothesis, the optimal $\bX'(1)^*$ and $\bX'(k)^*$ are given by the optimization problem \eqref{eq: opt rand new arrival}. Fix a user prompt length $Q_i$ and a model response length $A_i$ and we compare the cost-to-do under $\bX(1)$ and $\bX(k)$. 
    \begin{itemize}
        \item if conversation $i(1)$ arrives next, then one need to evict one more token from $\bX(1)$ than from $\bX(k)$.
        Suppose the extra token evicted from $\bX(1)$ is from conversation $i(k)$, then $\bX'(1)^* = \bX'(k)^*$. If not, then this means the extra token evicted is from another conversation with better ranking criterion, thus we have 
            \begin{align*}
            & J_{k+2}(\Gamma(\tilde{\blam}, i(1)), \Psi(\tilde{L}, i(1), Q_{i(1)}+A_{i(1)}), \bX'(1)^*) \\
            \leq & J_{k+2}(\Gamma(\tilde{\blam}, i(1)), \Psi(\tilde{L}, i(1), Q_{i(1)}+A_{i(1)}), \bX'(k)^*)
        \end{align*}
        by the induction hypothesis.
        \item if conversation $i(k)$ arrives next, then one need to evict one more token from $\bX(k)$ than from $\bX(1)$. In the optional caching model, the extra token evicted form $\bX(k)$ must be from conversation $i(1)$ by the definition of the ranking of conversations, thus $\bX'(1)^* = \bX'(k)^*$. 
        \begin{align*}
            & J_{k+2}(\Gamma(\tilde{\blam}, i(k)), \Psi(\tilde{L}, i(k), Q_{i(k)}+A_{i(k)}), \bX'(k)^*) \\
            = & J_{k+2}(\Gamma(\tilde{\blam}, i(k)), \Psi(\tilde{L}, i(k), Q_{i(k)}+A_{i(k)}), \bX'(k)^*)
        \end{align*}
        \item if conversation other than $i(1), i(k)$ arrives next, then $\bX'(k)^*$ and $\bX'(1)^*$ need to evict the same number of tokens. If $\bX'(k)$ evicts at least one token from conversation $i(k)$, then $\bX'(1)^* = \bX'(k)^*$. If not, then this means $\bX'(1)^*$ evicts one token from another conversation with better ranking criterion, thus we have 
        \[ J_{k+2}(\Gamma(\tilde{\blam}, i)), \Psi(\tilde{L}, i, A_{i}), \bX'(1)^*) \leq J_{k+2}(\Gamma(\tilde{\blam}, i), \Psi(\tilde{L}, i, A_{i}), \bX'(k)^*).\]
    \end{itemize}
\end{itemize}
Therefore, by induction, the result holds for all $k\geq 1$.
\end{proof}

\noindent\textbf{Greedy Implementation.}
The optimization problem \eqref{eq: opt rand new arrival} need not be solved explicitly as a token‑by‑token greedy procedure suffices. At a high-level, the policy ranks each token by arrival rates weighted by its counterfactual cost, i.e., the cost increase when we evict this token. 
\[\mathbb{P}(L_i + Q_i - \xi \geq X_i) = \mathbb{E}[(L_i + Q_i - \xi - (X_i - 1))^+] - \mathbb{E}[(L_i + Q_i - \xi - X_i)^+]\]
i.e., the difference in expected cost if we further evict one token when we have $X_i$ tokens in cache. 

\begin{algorithm}[!htbp]
\caption{Expected‑Tail‑Optimized LRU Policy}
\label{alg:expected-tail-optimized-lru}
\SetAlgoLined
\SetAlgoNoEnd
\DontPrintSemicolon
\KwIn{Number of conversations $N$, cache sizes $\{X_i\}$, current lengths $\{L_i\}$, belief turn rates $\{\lambda_i\}$, distribution of length of user prompt $\{Q_i\}$, threshold $\xi$, arriving conversation $\theta$, arriving user prompt length Q, arriving model response length A, tokens to evict $n$}
\KwOut{Updated cache sizes $\{X_i\}$}
$\mathrm{evicted}\gets 0$ \\
$L_\theta \gets L_\theta + Q + A$ \tcp*{Update system state for arriving conversation $\theta$}
$X_\theta \gets L_\theta$\\
$\lambda_\theta \gets \bar{\lambda}_\theta$\\
\For{each $i \in E$}{
        Compute $v_i \gets \lambda_i \cdot \mathbb{P}(L_i + Q_i - \xi \geq X_i)$ \tcp*{Ranking criterion}
    }
\While{$\mathrm{evicted} < n$}{
    Find $j = \arg\min_{i \in [N], X_i \geq 1} v_i$ \tcp*{Conversation with minimum value}
    $X_j \gets X_j - 1$ \tcp*{Evict one token}
        $\mathrm{evicted} \gets \mathrm{evicted} +1$ \\
    $v_j \gets \lambda_j \cdot \mathbb{P}(L_j + Q_j - \xi \geq X_j)$ \tcp*{Update ranking criterion}
}
\Return{$\{X_i\}$.}
\end{algorithm}
In the implementation of the policy, one can use min-heap to process which token to evict using the ranking criterion. The computational complexity of the policy is given by $\mathcal{O}(|E| + n \log{|E|})$, where $|E|$ is number of conversations with non-zero cached tokens and $n$ is the number of tokens one needs to evict. 

We show that policy \ref{alg:expected-tail-optimized-lru} indeed returns a cache state that is an optimal solution to the optimization problem \eqref{eq: opt rand new arrival}.

\begin{lemma}
    Policy \ref{alg:expected-tail-optimized-lru} returns an optimal solution to the optimization problem \eqref{eq: opt rand new arrival}.
\end{lemma}
\begin{proof}
We prove by contradiction. Note that it is possible for the policy to return multiple optimal solutions, and it is also possible for the optimization problem \eqref{eq: opt rand new arrival} to have multiple optimal solutions. Suppose not, then the two set of solutions do not intersect. Let $\bX^*$ denote one optimal solution. Then there must exist two conversations $i$, $j$ such that $X_j^* \geq 1$ and 
\[\lambda_i \mathbb{P}(L_i + Q_i - \xi \geq X_i^*+1) > \lambda_j \mathbb{P}(L_j + A_j - \xi \geq X_j^*),\]

Then we can construct another solution $\bX'$ such that $X'_k = X^*_k$ for $k \neq i,j$, and $X_i' = X_i^*+1, X_j' = X_j^*-1$. Then the difference between the objective values of $\bX'$ and $\bX^*$ is given by 
\begin{align*}
    \textsf{OBJ}(\bX')  - \textsf{OBJ}(\bX^*) & = \lambda_{i} \mathbb{E}[(L_{i} + Q_{i} - \xi - (X^*_{i}+1))^+] + \lambda_{j} \mathbb{E}[(L_j+ Q_j - \xi -( X^*_j-1))^+]\\
    & \quad \quad - \left(\lambda_{i} \mathbb{E}[(L_{i} + Q_{i} - \xi - X^*_{i})^+] + \lambda_{j} \mathbb{E}[(L_j+ Q_j - \xi -X_j^*)^+]\right) \\
    & = \lambda_{j} \mathbb{P}(L_{j} + A_{j} - \xi \geq X_{j}) - \lambda_{i} \mathbb{P}(L_{i} + Q_{i} - \xi \geq X_{i}+1)  \\
    & < 0, 
\end{align*}
which contradicts the optimality of $\bX^*$.
\end{proof}

\section{Discussion on Forced Caching}\label{apx: forced caching}
\noindent\textbf{Implementation of Tail-Optimized LRU.}
To implement forced caching, especially at GPU level, the server needs to decide which block to evict as serving the turn. The server may not know the total number of cache blocks to evict due to the uncertainty in the model response length, nevertheless the server can repeatedly call our policy to evict more tokens if needed.

\noindent\textbf{Hindsight optimal policy.}
To model forced caching, we replace optional caching constraint \eqref{cons: optional caching} with
\begin{align}\label{cons: forced caching}
    x_{i,t+1} = \sum_{j=1}^{t} (q_{i,j} + a_{i,j}), \forall i \in [N], t\in \mathcal{T}_i \text{ and } t < T,
\end{align} 
i.e., when a turn arrives, the server is required to cache its whole conversation history including newly generated response.
Theorem \ref{thm: modified baledy is optimal} continues to hold under forced caching.


\noindent\textbf{Expected-Tail-Optimized LRU.}
To model forced caching, we replace feasible caching decision space under optional caching with
        \[\mathcal{X}_{\mathcal{F}}(\bm{X}, \theta, L) = \{\bm{Y} \in \mathbb{N}^{\textsf{dim}(\bX, \theta)}: \sum_{i} Y_i \leq C,  Y_\theta = L, 0 \leq Y_i \leq X_i, i\neq \theta\}.\]
Theorem \ref{thm: optimality of tail LRU rand} continues to hold under forced caching, i.e., Expected-Tail-Optimized LRU remains to be optimal, if
\begin{itemize}
    \item every future prompt (if it arrives) has a known, fixed length $Q\ge0$. Here this fixed length can be heterogeneous across conversations and across turns. Crucially, the decision-maker still does not know if any given conversation will return; they only know that should it return, its next-turn question length will be $Q$. In this case, Expected-Tail-Optimized LRU is reduced to a deterministic version as stated in policy \ref{alg:tail-optimized-lru}.
    \item conversations have homogeneous turn rates $\lambda_{\textsf{turn}}$.
\end{itemize}
This fixed‑prompt‑length assumption holds when prompts are pre‑specified. When $Q=0$ after the first turn and responses also have zero length, our model reduces to classic paging with unit page size.

\section{Additional Figures}\label{sec:additional-figures}


\begin{figure}[!htbp]
    \centering
    \includegraphics[width=\linewidth]{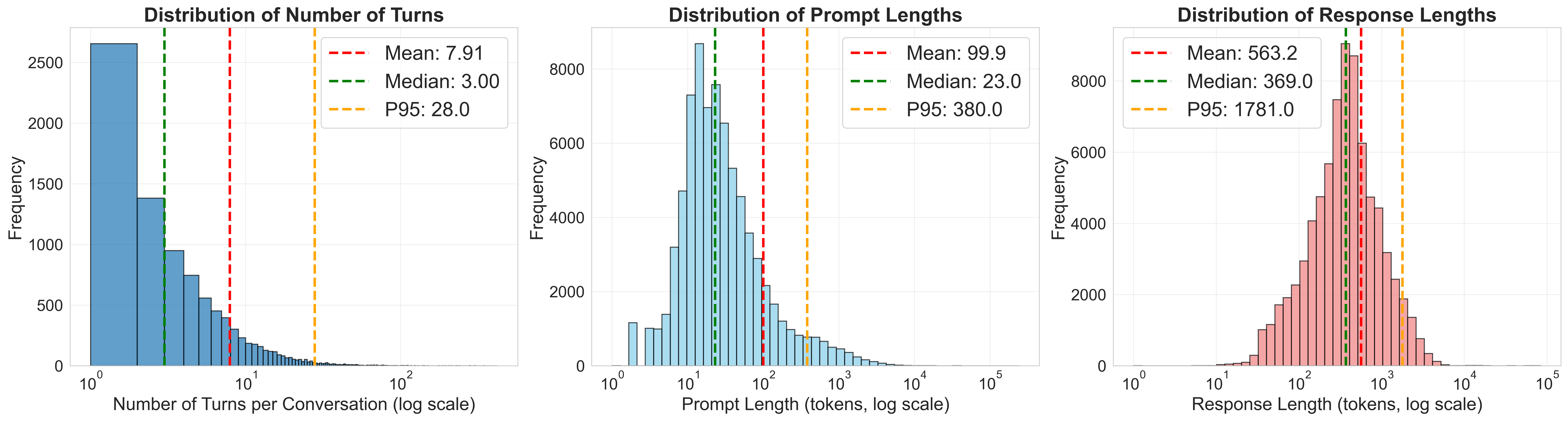}
    \caption{Distributions of turns and tokens of ShareGPT~\citep{sharegpt} datasets (sampled 10,000 conversations).}
    \label{fig: sharegpt distribution}
\end{figure}
\section{KV Cache Size Computation}\label{apx: KV cache size}

We calculate the KV cache memory requirements for the Vicuna-7B model. For 10,000 tokens stored in Float16 precision, the total KV cache size is approximately \textbf{4.88 GB}. With Float32 precision, this memory requirement doubles to approximately 9.77 GB.

\subsection{Model Configuration}
The following parameters are from the Vicuna-7B-v1.5 model configuration:\footnote{\url{https://huggingface.co/lmsys/vicuna-7b-v1.5/blob/main/config.json}}
\begin{itemize}
    \item Hidden Size: 4096
    \item Number of Attention Heads: 32
    \item Number of Hidden Layers: 32
    \item Number of Key-Value Heads: 32
    \item Head Size: 128 (Hidden Size $\div$ Number of Attention Heads)
    \item Data Type: Float16 (2 bytes per value)
\end{itemize}

\subsection{Calculation}
The KV cache size per token is computed as:
\begin{align*}
\text{KV cache per token (bytes)} &= 2 \times \text{Layers} \times \text{KV Heads} \times \text{Head Size} \times \text{Data Type Size} \\
&= 2 \times 32 \times 32 \times 128 \times 2 \\
&= 524{,}288 \text{ bytes}
\end{align*}
where the leading factor of 2 accounts for storing both Key (K) and Value (V) matrices.

For 10,000 tokens, the total memory requirement is:
\begin{align*}
\text{Total KV cache (bytes)} &= 10{,}000 \times 524{,}288 \\
&= 5{,}242{,}880{,}000 \text{ bytes}
\end{align*}

Converting to gigabytes:
\begin{align*}
\text{KV cache size (GB)} &= \frac{5{,}242{,}880{,}000}{1024^3} \\
&\approx 4.8828 \text{ GB}
\end{align*}

\section{Results on ShareGPT with Synthetic Timestamps}\label{apx: additional experiments}
ShareGPT~\citep{sharegpt} does not include timestamps of each request, thus we generate them with the stochastic model described in Section \ref{sec: stochastic model}. Specifically, for each conversation, we draw exponential inter-arrival times with rate $\lambda_{\textsf{conv}} = 1$, then for each conversation, we generate inter-arrival times between each turn within a conversation using Exponential distribution with rate $\lambda_{\textsf{turn}} = 3$. The average prompt length in ShareGPT is approximately 100 tokens (we thus set $\hat{Q} = 100$ in implementation), with an average of 3.5 turns per conversation.  

Tables \ref{tab:improvement-tlru-lru-sharegpt}–\ref{tab:improvement-tlru-thre-sharegpt} show that Tail-Optimized LRU still beats both LRU and Threshold-LRU: it trims P90 by up to 10\%, and P95 by up to 7\%. 
The smaller improvement compared to the ones observed in WildChat (Tables \ref{tab:improvement-tlru-lru}–\ref{tab:improvement-tlru-thre}) stem from the already-high base latencies under LRU (with capacity $C=1000$ under LRU, medium, P90, P95, P99 tail latencies are roughly 209 ms, 1415 ms, 2447 ms, 3649 ms), thus percentage improvements shrink. 

\begin{table}[!htp]\centering
\caption{Relative latency improvement of T-LRU over LRU with various $\xi_s$ (ShareGPT)}
\scriptsize
\begin{tabular}{rrrrrrrrrrr}\toprule
&\multicolumn{2}{c}{$\xi_s$ = 50ms} &\multicolumn{2}{c}{$\xi_s$ = 100ms} &\multicolumn{2}{c}{$\xi_s$ = 200ms} &\multicolumn{2}{c}{$\xi_s$ = 300ms} &\multicolumn{2}{c}{$\xi_s$ = 500ms}  \\\cmidrule{1-11}
Capacity &p90 &p95 &p90 &p95 &p90 &p95 &p90 &p95 &p90 &p95  \\\midrule
1000 &\cellcolor[HTML]{ffffff}0.0\% &\cellcolor[HTML]{ffffff}0.0\% &\cellcolor[HTML]{ffffff}0.0\% &\cellcolor[HTML]{ffffff}0.0\% &\cellcolor[HTML]{ffffff}0.0\% &\cellcolor[HTML]{ffffff}0.0\% &\cellcolor[HTML]{dcf1e6}0.9\% &\cellcolor[HTML]{dcf1e6}0.6\% &\cellcolor[HTML]{dcf1e6}0.9\% &\cellcolor[HTML]{dcf1e6}0.9\%  \\
2000 &\cellcolor[HTML]{dcf1e6}0.7\% &\cellcolor[HTML]{ffffff}0.0\% &\cellcolor[HTML]{dcf1e6}0.7\% &\cellcolor[HTML]{dcf1e6}0.2\% &\cellcolor[HTML]{dcf1e6}0.9\% &\cellcolor[HTML]{dcf1e6}1.5\% &\cellcolor[HTML]{dcf1e6}1.4\% &\cellcolor[HTML]{dcf1e6}2.0\% &\cellcolor[HTML]{c9e9d9}2.7\% &\cellcolor[HTML]{c9e9d9}2.3\%  \\
4000 &\cellcolor[HTML]{dcf1e6}0.6\% &\cellcolor[HTML]{dcf1e6}0.6\% &\cellcolor[HTML]{dcf1e6}0.7\% &\cellcolor[HTML]{dcf1e6}1.8\% &\cellcolor[HTML]{dcf1e6}2.0\% &\cellcolor[HTML]{c9e9d9}2.5\% &\cellcolor[HTML]{c9e9d9}2.5\% &\cellcolor[HTML]{c9e9d9}3.0\% &\cellcolor[HTML]{c9e9d9}4.2\% &\cellcolor[HTML]{c9e9d9}3.5\%  \\
6000 &\cellcolor[HTML]{dcf1e6}0.7\% &\cellcolor[HTML]{dcf1e6}1.3\% &\cellcolor[HTML]{c9e9d9}2.1\% &\cellcolor[HTML]{c9e9d9}2.9\% &\cellcolor[HTML]{c9e9d9}4.9\% &\cellcolor[HTML]{c9e9d9}3.6\% &\cellcolor[HTML]{b8e2ce}8.1\% &\cellcolor[HTML]{c9e9d9}4.2\% &\cellcolor[HTML]{96d4b6}10.0\% &\cellcolor[HTML]{c9e9d9}4.7\%  \\
8000 &\cellcolor[HTML]{dcf1e6}1.5\% &\cellcolor[HTML]{dcf1e6}0.9\% &\cellcolor[HTML]{c9e9d9}2.6\% &\cellcolor[HTML]{dcf1e6}1.8\% &\cellcolor[HTML]{c9e9d9}3.5\% &\cellcolor[HTML]{c9e9d9}2.5\% &\cellcolor[HTML]{c9e9d9}4.8\% &\cellcolor[HTML]{c9e9d9}3.6\% &\cellcolor[HTML]{b8e2ce}9.6\% &\cellcolor[HTML]{c9e9d9}5.0\%  \\
10000 &\cellcolor[HTML]{dcf1e6}0.9\% &\cellcolor[HTML]{dcf1e6}0.7\% &\cellcolor[HTML]{c9e9d9}3.6\% &\cellcolor[HTML]{dcf1e6}1.7\% &\cellcolor[HTML]{c9e9d9}4.3\% &\cellcolor[HTML]{c9e9d9}2.9\% &\cellcolor[HTML]{b8e2ce}5.1\% &\cellcolor[HTML]{c9e9d9}3.6\% &\cellcolor[HTML]{b8e2ce}9.0\% &\cellcolor[HTML]{b8e2ce}6.9\%  \\
\bottomrule
\end{tabular} \label{tab:improvement-tlru-lru-sharegpt}
\end{table}
\begin{table}[!htp]\centering
\caption{Relative latency improvement of T-LRU over Threshold-LRU with various $\xi_s$ (ShareGPT)}
\scriptsize
\begin{tabular}{rrrrrrrrrrr}\toprule
&\multicolumn{2}{c}{$\xi_s$ = 50ms} &\multicolumn{2}{c}{$\xi_s$ = 100ms} &\multicolumn{2}{c}{$\xi_s$ = 200ms} &\multicolumn{2}{c}{$\xi_s$ = 300ms} &\multicolumn{2}{c}{$\xi_s$ = 500ms}  \\\cmidrule{1-11}
Capacity &p90 &p95 &p90 &p95 &p90 &p95 &p90 &p95 &p90 &p95  \\\midrule
1000 &\cellcolor[HTML]{ffffff}0.0\% &\cellcolor[HTML]{ffffff}0.0\% &\cellcolor[HTML]{ffffff}0.0\% &\cellcolor[HTML]{ffffff}0.0\% &\cellcolor[HTML]{ffffff}0.0\% &\cellcolor[HTML]{ffffff}0.0\% &\cellcolor[HTML]{dcf1e6}0.9\% &\cellcolor[HTML]{dcf1e6}0.6\% &\cellcolor[HTML]{dcf1e6}0.9\% &\cellcolor[HTML]{dcf1e6}0.9\%  \\
2000 &\cellcolor[HTML]{dcf1e6}0.7\% &\cellcolor[HTML]{ffffff}0.0\% &\cellcolor[HTML]{dcf1e6}0.7\% &\cellcolor[HTML]{dcf1e6}0.2\% &\cellcolor[HTML]{dcf1e6}0.9\% &\cellcolor[HTML]{dcf1e6}1.5\% &\cellcolor[HTML]{dcf1e6}1.4\% &\cellcolor[HTML]{dcf1e6}2.0\% &\cellcolor[HTML]{c9e9d9}2.7\% &\cellcolor[HTML]{c9e9d9}2.3\%  \\
4000 &\cellcolor[HTML]{dcf1e6}0.6\% &\cellcolor[HTML]{dcf1e6}0.6\% &\cellcolor[HTML]{dcf1e6}0.7\% &\cellcolor[HTML]{dcf1e6}1.8\% &\cellcolor[HTML]{dcf1e6}2.0\% &\cellcolor[HTML]{c9e9d9}2.5\% &\cellcolor[HTML]{c9e9d9}2.5\% &\cellcolor[HTML]{c9e9d9}3.0\% &\cellcolor[HTML]{c9e9d9}4.2\% &\cellcolor[HTML]{c9e9d9}3.5\%  \\
6000 &\cellcolor[HTML]{dcf1e6}0.7\% &\cellcolor[HTML]{dcf1e6}0.8\% &\cellcolor[HTML]{c9e9d9}2.1\% &\cellcolor[HTML]{c9e9d9}2.4\% &\cellcolor[HTML]{c9e9d9}4.9\% &\cellcolor[HTML]{c9e9d9}3.1\% &\cellcolor[HTML]{b8e2ce}8.1\% &\cellcolor[HTML]{c9e9d9}3.7\% &\cellcolor[HTML]{96d4b6}10.0\% &\cellcolor[HTML]{c9e9d9}4.2\%  \\
8000 &\cellcolor[HTML]{dcf1e6}0.8\% &\cellcolor[HTML]{dcf1e6}0.3\% &\cellcolor[HTML]{dcf1e6}2.0\% &\cellcolor[HTML]{dcf1e6}1.2\% &\cellcolor[HTML]{c9e9d9}2.9\% &\cellcolor[HTML]{dcf1e6}2.0\% &\cellcolor[HTML]{c9e9d9}4.1\% &\cellcolor[HTML]{c9e9d9}3.0\% &\cellcolor[HTML]{b8e2ce}9.0\% &\cellcolor[HTML]{c9e9d9}4.5\%  \\
10000 &\cellcolor[HTML]{dcf1e6}0.9\% &\cellcolor[HTML]{dcf1e6}0.6\% &\cellcolor[HTML]{c9e9d9}3.6\% &\cellcolor[HTML]{dcf1e6}1.6\% &\cellcolor[HTML]{c9e9d9}4.3\% &\cellcolor[HTML]{c9e9d9}2.8\% &\cellcolor[HTML]{b8e2ce}5.1\% &\cellcolor[HTML]{c9e9d9}3.5\% &\cellcolor[HTML]{b8e2ce}9.0\% &\cellcolor[HTML]{b8e2ce}6.8\%  \\
\bottomrule
\end{tabular} \label{tab:improvement-tlru-thre-sharegpt}
\end{table}

Using a 200 ms SLO, T-LRU cuts the share of requests above the budget by 2–8\% across capacities (Table \ref{tab:merged-improvement-tlru-sharegpt}). Improvements again peak when $\xi_s$ is near the desired percentile; extremely high $\xi_s$ trades those mid-tail wins for heavier protection of the extreme tail, echoing the WildChat pattern.

\begin{table}[!htp]\centering
\caption{Relative improvement of T-LRU: \% reduction in requests with latency $>$ 200ms (ShareGPT)}
\scriptsize
\begin{tabular}{rrrrrrrrrrr}\toprule
Capacity &\multicolumn{2}{c}{$\xi_s$ = 50ms} &\multicolumn{2}{c}{$\xi_s$ = 100ms} &\multicolumn{2}{c}{$\xi_s$ = 150ms} &\multicolumn{2}{c}{$\xi_s$ = 200ms} &\multicolumn{2}{c}{$\xi_s$ = 500ms} \\
 &LRU &Thre-LRU &LRU &Thre-LRU &LRU &Thre-LRU &LRU &Thre-LRU &LRU &Thre-LRU \\\cmidrule{1-11}
1000 & \cellcolor[HTML]{dcf1e6} 0.7\% & \cellcolor[HTML]{ffffff} 0.0\% & \cellcolor[HTML]{dcf1e6} 1.7\% & \cellcolor[HTML]{dcf1e6} 1.0\% & \cellcolor[HTML]{dcf1e6} 2.3\% & \cellcolor[HTML]{dcf1e6} 1.6\% & \cellcolor[HTML]{dcf1e6} 2.3\% & \cellcolor[HTML]{dcf1e6} 1.6\% & \cellcolor[HTML]{ffd6d6} -3.5\% & \cellcolor[HTML]{ffd6d6} -4.2\% \\
2000 & \cellcolor[HTML]{dcf1e6} 1.1\% & \cellcolor[HTML]{dcf1e6} 1.0\% & \cellcolor[HTML]{dcf1e6} 1.9\% & \cellcolor[HTML]{dcf1e6} 1.8\% & \cellcolor[HTML]{dcf1e6} 2.6\% & \cellcolor[HTML]{dcf1e6} 2.5\% & \cellcolor[HTML]{dcf1e6} 3.1\% & \cellcolor[HTML]{dcf1e6} 3.0\% & \cellcolor[HTML]{ffd6d6} -8.6\% & \cellcolor[HTML]{ffd6d6} -8.7\% \\
4000 & \cellcolor[HTML]{dcf1e6} 1.8\% & \cellcolor[HTML]{dcf1e6} 1.3\% & \cellcolor[HTML]{dcf1e6} 3.9\% & \cellcolor[HTML]{dcf1e6} 3.4\% & \cellcolor[HTML]{dcf1e6} 4.7\% & \cellcolor[HTML]{dcf1e6} 4.2\% & \cellcolor[HTML]{dcf1e6} 4.8\% & \cellcolor[HTML]{dcf1e6} 4.3\% & \cellcolor[HTML]{ffd6d6} -15.6\% & \cellcolor[HTML]{ffd6d6} -16.2\% \\
6000 & \cellcolor[HTML]{dcf1e6} 1.9\% & \cellcolor[HTML]{dcf1e6} 1.1\% & \cellcolor[HTML]{dcf1e6} 4.7\% & \cellcolor[HTML]{dcf1e6} 3.9\% & \cellcolor[HTML]{c9e9d9} 6.0\% & \cellcolor[HTML]{c9e9d9} 5.2\% & \cellcolor[HTML]{dcf1e6} 4.7\% & \cellcolor[HTML]{dcf1e6} 3.9\% & \cellcolor[HTML]{ffd6d6} -26.2\% & \cellcolor[HTML]{ffd6d6} -27.2\% \\
8000 & \cellcolor[HTML]{dcf1e6} 1.2\% & \cellcolor[HTML]{dcf1e6} 0.9\% & \cellcolor[HTML]{dcf1e6} 4.3\% & \cellcolor[HTML]{dcf1e6} 4.0\% & \cellcolor[HTML]{dcf1e6} 4.9\% & \cellcolor[HTML]{dcf1e6} 4.6\% & \cellcolor[HTML]{dcf1e6} 2.9\% & \cellcolor[HTML]{dcf1e6} 2.6\% & \cellcolor[HTML]{ffd6d6} -39.3\% & \cellcolor[HTML]{ffd6d6} -39.7\% \\
10000 & \cellcolor[HTML]{dcf1e6} 2.2\% & \cellcolor[HTML]{dcf1e6} 1.8\% & \cellcolor[HTML]{c9e9d9} 6.5\% & \cellcolor[HTML]{c9e9d9} 6.1\% & \cellcolor[HTML]{c9e9d9} 7.9\% & \cellcolor[HTML]{c9e9d9} 7.6\% & \cellcolor[HTML]{dcf1e6} 3.3\% & \cellcolor[HTML]{dcf1e6} 3.0\% & \cellcolor[HTML]{ffd6d6} -50.7\% & \cellcolor[HTML]{ffd6d6} -51.2\% \\
\bottomrule
\end{tabular} \label{tab:merged-improvement-tlru-sharegpt}
\end{table}

Lastly, in spite of the extra foresight, End-Aware T-LRU and Length-Aware T-LRU show only marginal gains over T-LRU, and all three policies perform very closely to Tail-Optimized Belady, the optimal hindsight policy that minimizes the Tail Excess Latency. This is exactly what our stochastic model predicts: under Poisson arrivals, LRU’s recency order is already a near-perfect proxy for ``furthest in the future'', the rule the hindsight policy uses for eviction, so extra foresight offers diminishing returns.

\begin{figure}[ht]
    \centering
    \begin{subfigure}[t]{\textwidth} 
        \centering \includegraphics[width=\linewidth]{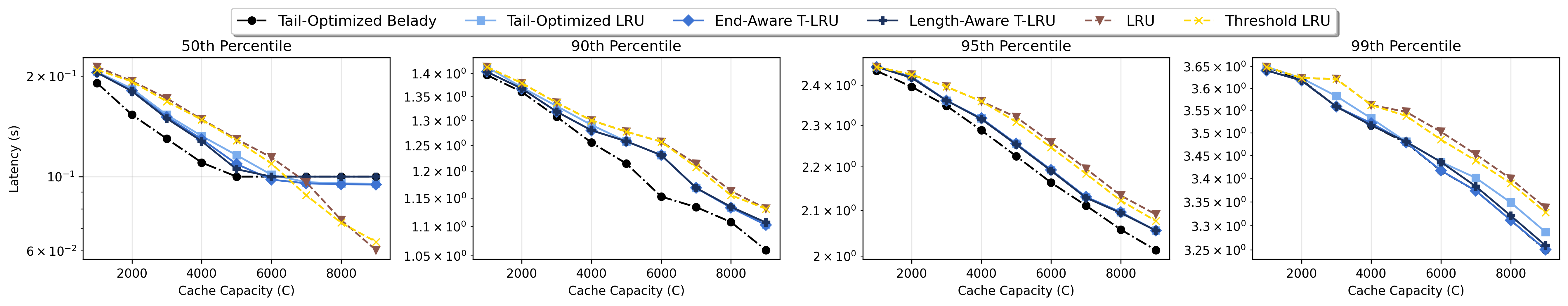}
    \end{subfigure} \hfill
    \begin{subfigure}[t]{\textwidth} 
        \centering \includegraphics[width=\linewidth]{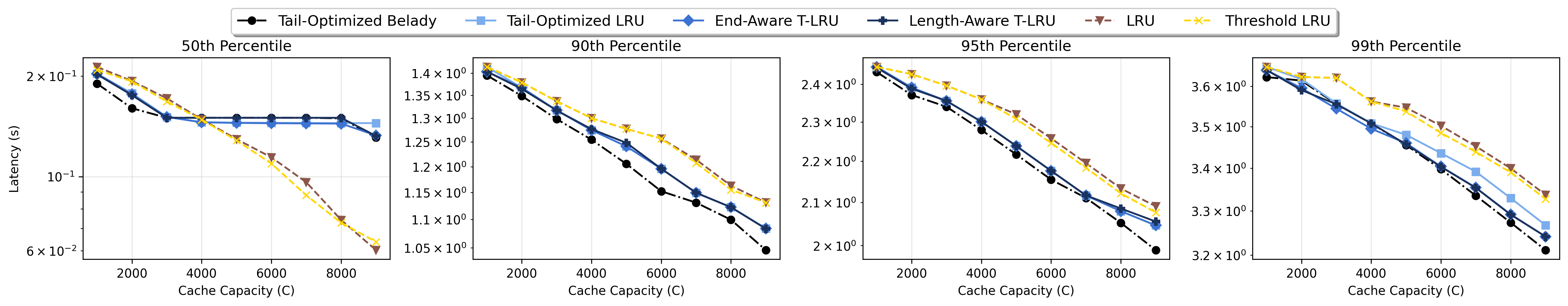}
    \end{subfigure}    
    \hfill
    \begin{subfigure}[t]{\textwidth} 
        \centering \includegraphics[width=\linewidth]{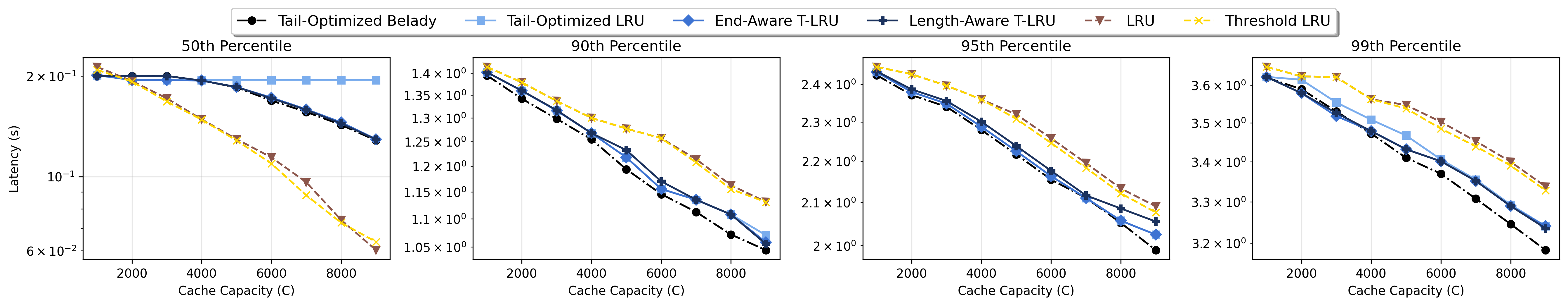}
    \end{subfigure}    
    \caption{Latency results for various settings (threshold latency $\xi_s$ = 100, 200, 300 ms) from top to bottom panels (ShareGPT)}
    \label{fig:result-latency-ShareGPT}
\end{figure}

\end{document}